\documentclass[12pt,onecolumn]{article}
 \usepackage{mathtools}

\usepackage{amsmath}  
 \usepackage{amssymb}
\usepackage{amsthm}
\usepackage[utf8x]{inputenc}
\usepackage{authblk}
\usepackage{url}
\usepackage{graphicx}
\usepackage{tikz}

\usepackage{caption}
\usepackage{subcaption}
\usepackage{setspace}
\usetikzlibrary{plotmarks}
\usepackage{pgfplots}

\usepackage{fullpage}

\newtheorem{lemma}{Lemma}
\newtheorem{proposition}{Proposition}
\newtheorem{assumption}{Assumption}
\DeclareGraphicsExtensions{.pdf,.png,.jpg,.jpeg,.eps}
\title{Energy Sharing for Multiple Sensor Nodes with Finite Buffers}
\author{Sindhu Padakandla}
\author{Prabuchandran K.J.}
\author{Shalabh Bhatnagar}

\affil{Dept of Computer Science and Automation, Indian Institute of Science\\E-Mail: $\{$sindhupr, prabu.kj, shalabh$\}$@csa.iisc.ernet.in}

\date{}
\allowdisplaybreaks

\pgfplotsset{
tick label style={font=\large},
}

\begin{document}
\newtheorem{remark}{\bf Remark}
\newtheorem{thm}{Theorem}
\newtheorem{mydef}{Definition}

\newcommand{\E}[1]{\,\mathbb{E}{\left[#1\right]}}
\newcommand{\R}{\text{ }\mathbb{R}}
\newcommand{\amin}[1]{\underset{#1}{\arg\min}}
\newcommand{\amax}[1]{\underset{#1}{\arg\max}}

\newcounter{mycounter}
\maketitle


\newcount\colveccount
\newcommand*\colvec[1]{
        \global\colveccount#1
        \begin{pmatrix}
        \colvecnext
}
\def\colvecnext#1{
        #1
        \global\advance\colveccount-1
        \ifnum\colveccount>0
                \\
                \expandafter\colvecnext
        \else
                \end{pmatrix}
        \fi
}

\begin{abstract} 
We consider the problem of finding optimal energy sharing policies that maximize the 
network performance of a system comprising of multiple sensor nodes and a single energy harvesting (EH) source. 
Sensor nodes periodically sense the random field and generate data, which is stored in the corresponding data queues.
The EH source harnesses energy from ambient energy sources and the generated energy is stored in an energy buffer.
Sensor nodes receive energy for data transmission from the EH source.
The EH source has to efficiently share the stored energy among the nodes in order to minimize the long-run average 
delay in data transmission. We formulate the problem of energy sharing between the nodes in the framework of average cost infinite-horizon Markov decision processes (MDPs).
We develop efficient energy sharing algorithms, namely Q-learning algorithm with exploration mechanisms based on the
$\epsilon$-greedy method as well as upper confidence bound (UCB).
We extend these algorithms by incorporating state and action space aggregation to tackle state-action space explosion in the MDP.
We also develop a cross entropy based method that incorporates policy parameterization in order to find near optimal energy sharing policies.
Through simulations, we show that our algorithms yield energy sharing policies that outperform the heuristic greedy method.
\end{abstract}
{\bf Keywords:}
\\ Energy harvesting sensor nodes, energy sharing, Markov decision process, Q-learning, state aggregation.

\section{Introduction}
\label{sec:intro}
A sensor network is a group of independent sensor nodes, each of which senses the environment.
Sensor networks find applications in weather and soil conditions monitoring, object tracking and structure monitoring.
Each sensor node in the network senses the environment and transmits the sensed data to a fusion node. 
The fusion node obtains data from several sensor nodes and carries out further processing. 

In order to sense the environment and transmit data to the fusion node, nodes require energy
and most often the nodes are equipped with pre-charged batteries for this purpose.
However, as the nodes exhaust their battery power and stop sensing, the network performance degrades.
The lifetime of the network is linked to the lifetimes of the individual nodes. 
Hence, the network becomes inoperable when a large number of nodes stop sensing. 
Thus, in a network with battery operated sensor nodes, the primary intention is to enhance the lifetime of the network,
which may often lead to a compromise in the network performance.
Many techniques have been proposed, which
focus on improving lifetime of networks of sensor nodes. 
One of the more recent techniques which deals with this problem is the usage of 
energy harvesting to provide a perpetual source of energy for the nodes.

An energy harvesting (EH) sensor node replenishes the energy it consumes by harvesting energy 
from the environment (e.g., solar, wind power etc.) or other sources (e.g., body movements, finger strokes etc.) 
and converting into electrical energy.
This way an EH node can be constantly powered through energy replenishment.
So when compared to networks consisting of battery operated nodes, the long-term network 
performance metrics become appropriate. Thus, the goal pertaining to an EH sensor network is to reduce the average delay in data transmission. 
Even though an EH sensor node potentially has infinite amount of energy, yet the energy harvested is infrequently available as
it is usually location and time dependent. Moreover the amount of energy replenished might be lower than the required amount.
Therefore it is important to match the energy consumption with the amount of energy harvested in order to prevent energy starvation.
This underlines the need for intelligently managing harvested energy to achieve the goal of good
network performance.

A drawback associated with an EH sensor (node) is that it requires additional circuitry to harvest energy, which increases the cost
of the node. A network which contains several such nodes is not economically viable. 
The cost of the network can be minimized if there exists a central EH source 
which harvests energy and shares the available energy among multiple sensor nodes in its vicinity.
Such an architecture is incorporated in \emph{motes}. A \emph{mote} (Fig. \ref{mote-pic}) is a single unit
on which sensors with different functionalities are arranged (see \cite{gay2010ultra}). For instance,
there could be pressure sensors, temperature sensors etc., in the same
unit to make different sets of measurements simultaneously. 
Alternatively, the sensors could be of the same functionality but deployed together
at different angles in order to have a $360^{\circ}$ view of the entire sensing
region. 

\begin{figure}[ht]
\begin{center}
\begin{tikzpicture}
 \node (mote) at (0,0) {\includegraphics[scale=0.2]{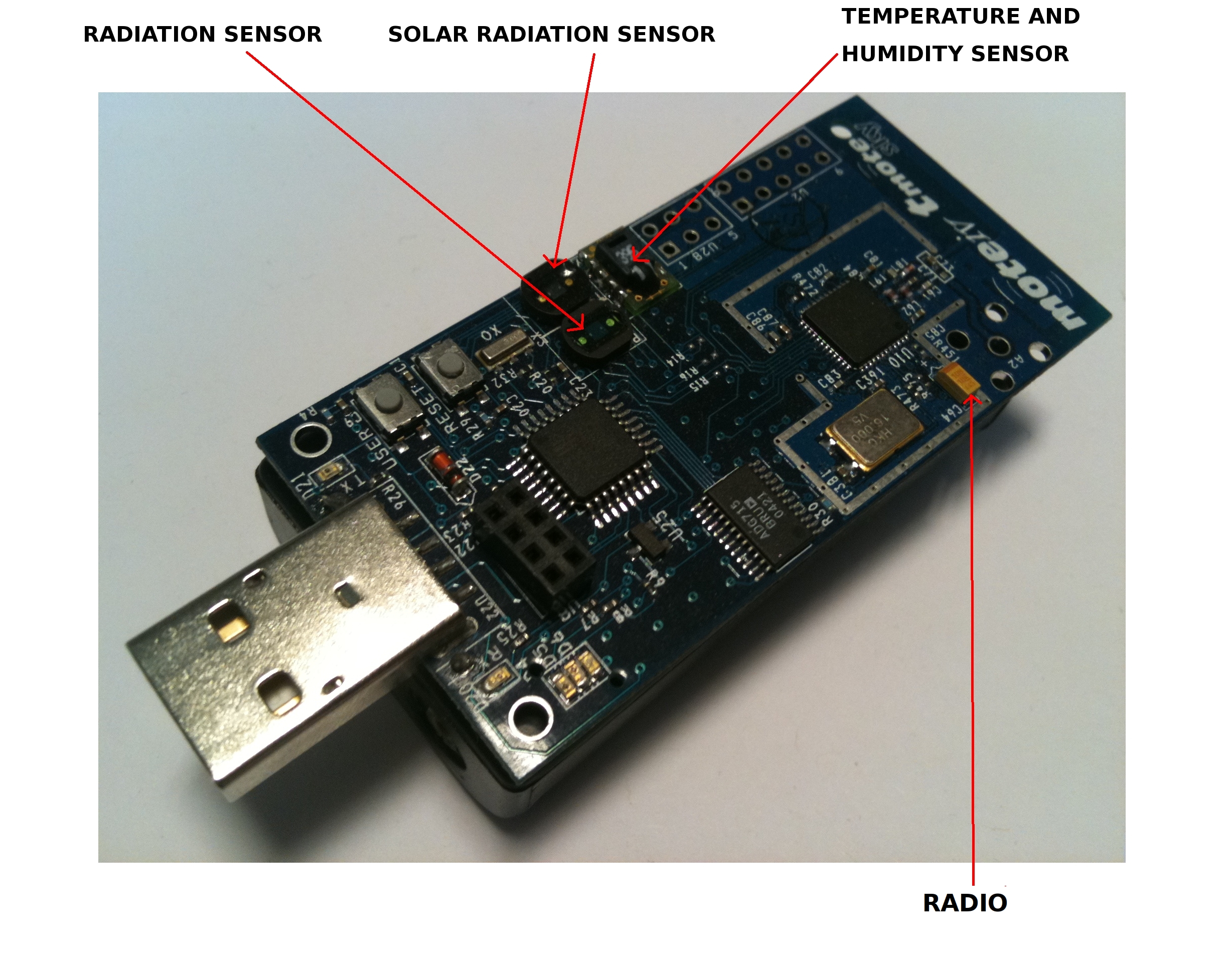}};
\end{tikzpicture}
\caption{Mote with pressure, humidity and temperature sensors (Courtesy: Advanticsys Pvt Ltd. and UC Berkeley)}
\label{mote-pic}
\end{center}
\end{figure}

Each of these sensors (within a unit) have their own data buffers
and a common EH source feeds energy to each of the data queues.
Usually, the EH source is a battery which is recharged by energy harvesting.
The sensors in the mote are perpetually powered, but only if the energy harvested in the source is efficiently shared. 
Thus there is a need for a technique that dynamically allocates
energy to each of the data buffers of individual sensors in order that the
average queue lengths (or transmission delays) across the data buffers are
minimized.

In this paper, we focus on the problem of developing algorithms that achieve
efficient energy allocation in a system comprising of multiple sensor nodes with their own 
data buffers and a common EH source. 
Another scenario (that however we do not consider here)
where our techniques are applicable is the case of downlink transmissions \cite{lalitha2013power},
where a base station (BS) maintains a separate data queue for each
individual sensor node. The BS in question would also typically be powered
by a single EH source, and again the problem would be to dynamically
allocate the available energy to each one of the data queues.
As suggested by a reviewer of the journal version of this paper, the above is equivalent to a communication
setup with with an energy harvesting transmitter and $n$ receivers which are connected to the transmitter over orthogonal
links and equal gain links. The transmitter employs $n$ finite data buffers to store incoming data,
intended for the $n$ receivers and must optimally allocate its energy to transmit data intended for the $n$ receivers.

We present learning algorithms for a controller which has to judiciously distribute the energy amongst the competing nodes.
The controller decides on the amount of energy to be allocated to every node at every decision instant considering the amount of data
waiting to be transmitted in each of the data queues. Thus the state of the system comprises of the amount of data in each of the data queues along with the energy available in the source. 
Given the system state at an instant, the controller has to find out the best possible way to allocate energy to the individual nodes.
The decided allocation has a bearing on the total amount of data transmitted at that instant 
as well as the amount of data that will be transmitted in the future.
Our algorithms help the controller learn the optimal allocation for every state, one which reduces the buildup of data in the data buffers.
In the algorithm we present, the controller systematically tries out several possible allocations feasible in a state,
before learning the optimal allocation. This method is computationally efficient for small number of states. 
However it becomes computationally expensive when there are numerous states. 
We propose approximation algorithms to find the near-optimal allocation of energy in this scenario.
In the following subsection, we survey literature on EH nodes and energy management policies employed in EH sensor networks.

\subsection{Related Work}
Optimizing energy usage in battery-powered sensors is addressed in \cite{6517237,ren2014dynamic}.
The problem of designing appropriate sensor schedules of sensor data transmission is discussed in \cite{6517237}.
A schedule of data transmission indicates when the battery-powered sensor transmits data. 
Transmitting data uses up energy, while not transmitting data results in error in estimation of 
parameters dependent on the sensor data.
The authors in \cite{6517237} consider battery-powered sensor nodes, each of which needs to 
minimize the energy utilized for data transmission. 
The estimation of parameters dependent on the sensor data may however involve error if the sensor does not
transmit data for long periods of time. 
The objective in \cite{6517237} is to find optimal periodic sensor schedules 
which minimize the estimation error at the fusion node and optimize energy usage.

In \cite{ren2014dynamic}, the authors consider battery-powered sensors with two 
transmission power levels. The transmission power levels have different packet drop rates with the higher transmission power level 
having a lower packet drop rate. The sensor can choose one of the power levels for data transmission. 
It is assumed that the fusion node sends an acknowledgment (ACK or NACK) to the sensor node which indicates whether
the data packet has been received or not. The objective in \cite{ren2014dynamic} is to minimize the average expected error in state estimation
under energy constraint. At time $k$, based on the communication feedback the sensor knows whether 
the previous packets have been received by the fusion node or not. The problem of choosing the transmission power level is modeled
as a MDP and the optimal schedule is shown to be stationary. The works \cite{6517237,ren2014dynamic} consider the problem of efficient energy usage 
in battery powered sensors. The aspect of network performance is not considered in these. Our work deals with optimizing energy sharing in EH nodes
where maximizing a network performance objective is the primary goal.

An early work in rechargeable sensors is \cite{kansal2003environmental}.
The authors of \cite{kansal2003environmental} present a framework for the sensor network to adaptively learn the spatio-temporal 
characteristics of energy availability and provide algorithms to use this information for task sharing among nodes. 
In \cite{kansal2007power}, the irregular and spatio-temporal characteristics of harvested energy are considered.
The authors discuss the conditions for ensuring \emph{energy-neutral} operation, i.e., using the energy harvested at an appropriate rate 
such that the system continues to operate forever. Practical methods for a harvesting system to achieve energy-neutral operation are developed.
Compared to \cite{kansal2003environmental,kansal2007power}, we focus on minimizing the delay in data transmission from the nodes
and also ensuring energy neutral operation.

The scenario of a single EH transmitter with limited battery capacity is considered in \cite{tutuncuoglu2011short,5992841}.
In \cite{5992841}, the transmitter communicates in a fading channel, whereas in \cite{tutuncuoglu2011short},
no specific constraints on the channel are considered.
The problem of finding the optimal transmission policy to maximize the short-term throughput of an EH transmitter is considered
in \cite{tutuncuoglu2011short}. Under the assumption of an increasing concave power-rate relationship, 
the short-term throughput maximizing transmission policy is identified.
In \cite{5992841}, the transmitter gets channel state information and the node has to adaptively control the transmission rate.
The objective is to maximize the throughput by a deadline and minimize the transmission completion time of a communication session.
The authors in \cite{5992841} develop an online algorithm which determines the transmit power at
every instant by taking into account the amount of energy available and channel state.

The efficient usage of energy in a single EH node has been dealt with in some recent works 
\cite{niyato2007wireless,sharma2010optimal,prabuchandran2013q,6094139}.
A channel and data queue aware sleep/active/listen mechanism in this direction is proposed in \cite{niyato2007wireless}.
Listen mode turns off the transmitter, while sleep mode is activated if channel quality is bad. The node periodically enters the active mode.
In the listen mode, the queue can build up resulting in packets being dropped. In the sleep mode, incoming packets are blocked. 
A bargaining game approach is used to balance the probabilities of packet drop and packets being blocked.
The Nash equilibrium solution of the game 
controls the sleep/active mode duration and the amount of energy used. 

The model proposed in \cite{sharma2010optimal,prabuchandran2013q} considers a single EH sensor node with
finite energy and data buffers. The authors assume that data sensed is independent across time instants and so is the energy harvested.
The amount of data that can be transmitted using some specified energy is modeled using a \emph{conversion function}.
In \cite{sharma2010optimal}, a linear conversion function is used and optimal energy management policies are provided for the same. 
These policies are throughput optimal and mean delay optimal in a low SNR regime. 
However, in the case of non-linear conversion function, \cite{sharma2010optimal} provides certain heuristic policies.
In \cite{prabuchandran2013q}, a non-linear conversion function is used. The authors therein provide simulation-based learning algorithms 
for the energy management problem. These algorithms are model-free, i.e., do not require an explicit model of the system and the conversion function.
Unlike \cite{tutuncuoglu2011short,5992841,sharma2010optimal,niyato2007wireless,prabuchandran2013q},
our work deals with multiple sensors sharing a common EH power source. The maximization objective is the delay in data transmission
from the nodes. However, channel constraints are not addressed in our work.

Data packet scheduling problems in EH sensor networks are considered in \cite{6094139} and \cite{6253062}. 
It is assumed in \cite{6094139} that a single EH node has separate data and energy queues, 
while the data sensed and energy harvested are random. The same assumption is made for each sensor in a two-sensor communication
system considered in \cite{6253062}.
For simplicity it is assumed that all data bits have arrived in the queue and are ready for transmission, while the energy harvesting times and harvested energy amounts are known before the transmission begins.
In \cite{6094139}(\cite{6253062}) the objective is to minimize the time by which all data packets from the node(s) 
are transmitted (to the fusion node). It is proposed to optimize this by controlling the transmission rate.
The authors develop an algorithm to find the transmission rate at every instant, which optimizes the time 
to transmit the data packets.
A two-user Gaussian interference channel with two EH sensor nodes and receivers is considered in \cite{tutuncuoglu2012sum}.
This paper focuses on short-term sum throughput maximization of data transmitted from the two nodes before a given deadline.
The authors provide generalized water-filling algorithms for the same.
In contrast to the models developed in \cite{6094139,6253062,tutuncuoglu2012sum}, our model assumes multiple sensors sharing a 
common energy source. The data and energy arrivals are uncertain and unknown. 
Moreover the problem we deal with has an infinite horizon, wherein the objective is to reduce the mean delay of
data transmission from the nodes. We develop simulation based learning algorithms for this problem.

Cooperative wireless network settings are considered in \cite{6702854,6657835,tutuncuoglu2013cooperative}.
Three different network settings with energy transfer between nodes are considered in \cite{6657835}.
Energy management policies which maximize the system throughput within a given duration are determined 
in all the three cases. A water-filling algorithm is developed which controls the flow of harvested 
energy over time and among the nodes.
In \cite{tutuncuoglu2013cooperative}, there exists an EH relay node and multiple other EH source nodes.
The source nodes have infinite data buffer capacity. The relay node transfers data between the source and destination nodes.
The source and relay nodes can transfer energy to one another. A sum rate maximization problem in this setting is solved. 
In \cite{6702854}, multiple pairs of sources and destinations communicate via an EH relay node. 
The EH relay node has a limited battery, which is recharged by wireless energy transfer from the source nodes.
The EH relay node has to efficiently distribute the power obtained among the multiple users. 
The authors investigate four different power allocation strategies for outage performance 
(outage is an event in which data is lost due to lack of battery energy or transmission failures caused by channel fades).
We do not consider energy cooperation between nodes in the sensor network. Moreover, we do not assume
wireless energy transfer in our model.

A multi-user additive white Gaussian noise (AWGN) broadcast channel comprising of a single EH transmitter and $M$ receivers is considered in 
\cite{ozel2012optimal}. The EH transmitter harvests energy from the environment and stores in a queue.
The transmitter has $M$ data queues, each of which stores data packets intended for a specific receiver.
The data queues have fixed number of bits to be delivered to the receiver.
The objective in \cite{ozel2012optimal} is to find a transmission policy that minimizes the time by 
which all the bits are transmitted to the receivers.
An optimization problem is formulated and structural properties of the optimal policy are derived.
In our work, we model energy sharing in multiple nodes when there is a single power source. 
We assume uncertain data and energy arrival processes. The objective is to minimize the average delay in data transmission
from the nodes, when there is data arrival at every instant.

\subsection{Our Contributions}
\begin{itemize} 
 \item We consider the problem of efficient energy allocation in a system with multiple
sensor nodes, each with its own data buffer, and a common EH source.
 \item We model the above problem as an infinite-horizon average cost Markov decision process (MDP)
\cite{Bertsekas:2007:DPO:1396348},\cite{puterman1994markov} with an appropriate single-stage cost function. 
Our objective in the MDP setting is to minimize the long-run average delay in data transmission.
\item We develop reinforcement learning algorithms which provide optimal energy sharing policies for the above problem. 
The learning procedure used does not need the system knowledge such as data and energy rates or 
cost structure and learns using the data obtained in an online manner.
 \item In order to deal with the dimensionality of the state space of the MDP, we present approximation
algorithms. These algorithms find near-optimal energy distribution profiles when the state-action space of the MDP becomes unmanageable.
 \item We demonstrate through simulations that the policies obtained from our algorithm are better than the policies
obtained from a heuristic greedy method and a combined nodes Q-learning algorithm (see Section \ref{sec:results}).
\end{itemize}

\subsection{Organization of the Paper}
The rest of the paper is organized as follows. The next section describes the model, related notation and assumptions. 
Section \ref{sec:es-prob-as-an-mdp} formulates the energy sharing problem as an MDP. 
Section \ref{sec:2nodes-algo} presents the RL algorithms used for solving the MDP. 
Section \ref{sec:approximation} highlights the need for approximate policies and gives a detailed explanation of the 
approximation algorithms we develop for the problem.
Section \ref{sec:results} presents the simulation results of our algorithms.
Section \ref{sec:conclusions} provides the concluding remarks and possible future directions.
Finally, an appendix at the end of the paper contains the proof of two results.
\section{Model and Notation}
\label{sec:model-notation}
We consider the problem of sharing the energy available in an energy harvesting source among multiple sensor nodes.
We present a slotted, discrete-time, model (Fig. \ref{fig:model}) for this problem.
A sensor node in the network senses a random field and stores the sensed data in a finite data buffer of size $D_{\!_{MAX}}$.
In order to transmit the sensed data to a fusion (or central) node, the sensor node needs energy, which it obtains
from an energy harvesting source. The energy harvesting source has an energy buffer of finite capacity $E_{\!_{MAX}}$.
The common EH source is an abtract entity in the model. It is generally a rechargeable battery which
is replenished by random energy harvests.
We assume fragmentation of data packets (fluid model) as in \cite{sharma2010optimal} and hence these will be treated as bit strings.
\begin{figure}[h]
\centering
\scalebox{0.6}{\input{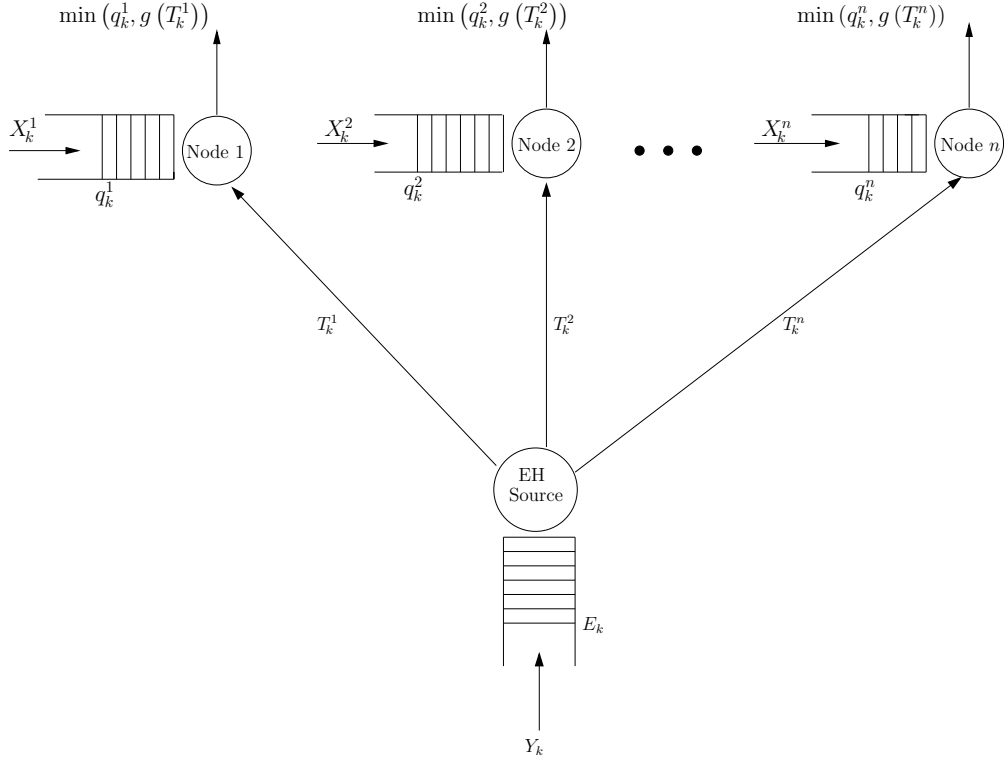}}
   \caption{The System Model}
   \label{fig:model}
\end{figure}

Let $q_{k}^{i}$ denote the data buffer level of node $i$ and $E_k$ be the energy buffer level at the beginning of slot $k$.
Sensor node $i$ generates $X_{\!k}^{i}$ bits of data by sensing the random field. The source harvests $Y_k$ units of energy.
Based on the data queue levels $(q^1_k,\ldots,q^n_k)$ and the energy level $E_k$, the energy sharing controller decides upon the number of energy
bits to be provided to every node. Let $T_{\!k}^{i}$ units of energy be provided to node $i$ in slot $k$. 
Using it, the node transmits $g(T_{\!k}^{i})$ bits of data. We have assumed the function $g$ to be monotonically non-decreasing and concave as with other references (\cite{sharma2010optimal,yang2010transmission,ho2010optimal,ozel2011adaptive,goyal2003power}). 
Note that the Shannon Channel capacity for Gaussian channels gives such a conversion function and in particular,
\begin{displaymath}
 g(T_k) = \frac{1}{2} \log(1 + \beta T_k),
\end{displaymath}
where $\beta$ is a constant and $\beta T_k$ gives the Signal-to-Noise (SNR) ratio. This is a non-decreasing concave function. 
We have assumed this form in the simulation experiments. However, our algorithms work regardless of the 
form of the conversion function and will learn the optimal energy sharing policy 
for any form of conversion function (see Remark \ref{remark:different-g}).

It should be noted that we do not consider wireless energy transfer from the source node to the sensor nodes.
Here we consider the source node to be a rechargeable battery which powers the nodes.
The queue lengths in the data buffers evolve with time as follows:
\begin{equation}
\label{eqn:qk-evolution}
q_{{k+1}}^{i} = ( q_{k}^{i} - g( T_{\!k}^{i} ) )^{+} + X_{\!{k}}^{i}	\qquad 1 \leq i \leq n, k \ge 0,
\end{equation}
where $(q^i_{k} - g( T^{i}_{\!k} ))^{+} = \max(q^i_{k} - g(T^{i}_{\!k}),0)$ and the energy buffer queue length evolves as given below:
\begin{equation}
\label{eqn:ek-evolution}
 E_{k+1} = \biggl(E_k - \sum_{i=1}^{n} \; T_{\!k}^{i} \biggr ) + Y_{\!{k}},	\qquad 1 \leq i \leq n, k \ge 0,
\end{equation}
where ${\displaystyle \sum_{i=1}^{n} } \; T_{\!k}^{i} \; \leq \: E_k $. 
\begin{assumption}
The generated data rates at time $k+1, X_{k+1} \triangleq (X_{k+1}^1, X_{k+1}^2, \ldots, X_{k+1}^n)$ where $n$ denotes the number of 
sensors in a node, evolves as a jointly Markov process, i.e.,
\begin{equation}
\label{eqn:evol1}
X_{k+1}=f^1(X_k,W_k), \qquad k \geq 0
\end{equation}
where $f^1$ is some arbitrary vector valued function with $n$ components
and $\{W_k,k\geq 1\}$ is a noise sequence with probability distribution $P(W_k\mid X_k)$ depending on $X_k$. 
Thus, the generated data $\{X_k, k \geq 0\}$ is both spatially and temporally correlated.
Moreover, the sequence $X_{\!k}^{i}, k \geq 0$ satisfies ${\sup\limits_{k \geq 0} }\; \E {X_{\!k}^{i} }\leq r < \infty$.
Further, the energy arrival process evolves as:
\begin{equation}
\label{eqn:evol2}
Y_{k+1}=f^2(Y_k,V_k), \qquad k \geq 0,
\end{equation}
where $f^2$ is some scalar valued function and $\{V_k,k\geq 1\}$ is the
noise sequence with probability distribution $P(V_k\mid Y_k)$ depending on $Y_k$.
\end{assumption}
\begin{remark}Assumption 1 is general enough to cover most of the stochastic models for the data and energy arrivals.
A special case of Assumption 1 is to consider that for any 
$k \geq 0 \; \text{and} \; 1\leq i\leq n$, $X^i_k$ is independent of 
$X_{\!{k-1}}^{i},X_{\!{k-2}}^{i},\ldots,X_{\!1}^{i},X_{\!0}^{i}$
and the given sequence $\{X^i_k\}_{k \geq 0}$ for a given $i \in \{1,\ldots,n\}$ is identically
distributed.
Similarly, for any $k \geq 0$, $Y_k$ is independent of $Y_{k-1},Y_{k-2},\ldots,Y_1,Y_0$ 
and the sequence $\{Y_k\}$ is identically distributed. In Section \ref{sec:results}
we show results of experiments where the above i.i.d setting as well as a more
general setting as described earlier are shown.
\end{remark}
\section{Energy Sharing Problem as an MDP}
\label{sec:es-prob-as-an-mdp}
A Markov decision process (MDP) is a tuple of states, actions, transition probabilities and single-stage costs. Given that
the MDP is in a certain state, and an action is chosen by the controller, the MDP moves to a `next' state 
according to the prescribed transition probabilities. The objective of the controller is to select a sequence
of actions as a function of the states in order to minimize a given long-term objective (cost).
We formulate the energy sharing problem in the MDP setting using the long-run average cost criterion.
The MDP formulation requires that we identify the states, actions and the cost structure
for the problem, which is described next.

The state $s_k$ is a tuple comprising of the data buffer level of all sensor nodes, the level of the energy buffer in
the source, the data and energy arrivals in the past. Note that for $1 \leq i \leq n, \; q_{k}^{i} \in \{0,1,\ldots,D_{\!_{MAX}}\}$. Similarly $E_k \in \{0,1,\ldots,E_{\!_{MAX}}\}$.
Thus in stage $k$, in the context of Assumption 1, state 
$s_k = (q_{k}^1,q_{k}^2,\ldots,q_{k}^n,E_k,X_{k-1},Y_{k-1})$. 
However, when we assume that for $1 \leq i \leq n$, $\{X^i_k\}$ and 
$\{Y_k\}$ are i.i.d (as in Remark 1), then the state tuple simplifies to $s_k = (q_{k}^1,q_{k}^2,\ldots,q_{k}^n,E_k)$.

The set of all states is the state-space, which is denoted by 
$S$. Similarly $A$ denotes the action-space, which is the set of all actions. The set of feasible actions in a state $s_k$ is denoted by $A(s_k)$.
A deterministic policy $\pi = \{T_{\!k}, k \geq 0\}$ is a sequence of maps such that at time $k$ when state 
$s_k = (q_{k}^1,\ldots,q_{k}^n,E_k,X_{k-1},Y_{k-1})$, 
i.e., when there are $q_{k}^j$ units of data at node $j$, $1 \leq j \leq n$ and $E_k$ bits of energy in the source, 
$X_k$ is the data arrival vector and $Y_K$ is the energy harvested at time $k-1$,
then $T_{k}(s_k) = (T_{\!k}^1(s_k),T_{\!k}^2(s_k),\ldots,T_{\!k}^n(s_k))$ gives the number of energy bits to be given 
to each node at time $k$ (i.e., it gives the energy split). 
Thus the action to be taken in state $s_k$ is given by $T_k(s_k) \in A(s_k)$.
A deterministic policy which does not change with time is referred to as a stationary deterministic policy (SDP). We denote such a policy $\pi$ as 
$\pi = (T,T,\ldots)$, where $T(s_k)$ is the action chosen in state $s_k$.
We set the single-stage cost $\tilde{c}(s_k,T(s_k))$ as a sum of the number of bits in the data buffers. Thus,
\begin{equation}
\label{eqn:cost-structure-not-used}
\tilde{c}(s_k,T(s_k)) = {\sum_{i = 1}^{n}} q^i_{k}.
\end{equation}

\begin{remark}
 In order to formulate the energy sharing problem in the framework of MDP,
we require the state sequence $\{s_k = (q_k^1,q_k^2,\ldots,q_k^n,E_k)\}_{k\geq0}$ under a given policy to be a
Markov chain, i.e.,
$$P(s_{k+1} \mid s_k, s_{k-1}, \ldots, s_0,\pi)=P(s_{k+1}\mid s_k,\pi).$$
We have generalized the assumption on $\{X_k^i, 1 \leq i \leq n\}_{k \geq 0}$ and $\{Y_k\}_{k \geq 0}$ and
consider jointly Markov data arrival and Markovian energy arrival processes.
Remark 1 applies to the i.i.d case.
If we assume the data arrivals  $\{X_k^i\}_{k \geq 0}$ for a fixed $i \in \{1, 2, \ldots, n\}$ and the energy arrivals
$\{Y_k\}_{k \geq 0}$ are i.i.d, then the Markov assumption can be seen to be easily
satisfied.

The Markov property for the state evolution $\{s_k\}_{k \geq 0}$ is necessary as we can only search for
policies based only on the present state of the system. Otherwise, the policies will be based on the entire
history. The search for optimal policies in the space of
history based policies is a computationally infeasible task.

In the general case where $\{X_k\}$ is jointly Markov, note that the state
sequence $\{s_k\}_{k \geq 0}$ under a given policy will not be a Markov chain.
Now consider the augmented state $\bar{s}_k \stackrel{\Delta}{=} \colvec{3}{s_k}{X_{k-1}}{Y_{k-1}}$. Now, under a given policy $\pi=(T,\ldots,T)$, the state evolution can be described as 
\begin{align}
\colvec{6}{q_{k+1}^1}{\vdots}{q_{k+1}^n}{E_{k+1}}{X_k}{Y_k} = \colvec{6}{(q_{k}^{1} - g( T^{1}(s_k) ) )^{+} + X_{\!{k}}^{1}\qquad}{\vdots}{(q_{k}^{n} - g( T^{n}(s_k) ) )^{+} + X_{\!{k}}^{n}\qquad}{ \bigl(E_k - \sum_{i=1}^{n} \; T^{i}(s_k) \bigr) + Y_{\!{k}},\qquad}{f^1(X_{k-1},W_{k-1})}{f^2(Y_{k-1},V_{k-1})}.
\end{align}
This can be written as
$\bar{s}_{k+1} = h(\bar{s}_k,T(\bar{s}_k),W_{k-1},V_{k-1})$ for suitable vector valued function $h$. This is the standard description for the state evolution for an MDP (see Chapter 1 in \cite{bertsekas1995dynamic}). Since the probability distribution of the noise $W_{k-1}$ ($V_{k-1}$) depends only on $X_{k-1}$ ($Y_{k-1}$), the augmented state sequence
$\bar{s}_{k}=\{(s_k,X_{k-1},Y_{k-1})\}_{k\geq0}$ forms a Markov chain. This facilitates search for policies only based on the
present augmented state. 
\end{remark}
\begin{remark}
The sensor node may generate data as packets, but in the model we allow for arbitrary fragmentation of data during transmission.
Hence packet boundaries are no longer relevant and we consider bit strings. This is the fluid model as described in 
\cite{fluidModelReference}. The data is considered to be stored in the data buffers as bit strings and hence the data buffer
levels are discrete. The fluid model assumption (data discretization) has been made in 
\cite{sharma2010optimal,goyal2003power,6094139}.
For energy harvesting we consider energy discretization.
Energy discretization implies that we have assumed that discrete levels of energy are harvested and stored
in the queue.
Energy discretization has been considered in some previous works \cite{aprem2013transmit,sharma2010optimal}.
Owing to these assumptions on data generation and energy harvesting, the state space is discrete and finite.
\end{remark}

The long-run average cost of an SDP $\pi$ is given by
\begin{equation}
\label{eqn:avgcost1}
\tilde{\lambda}^{\pi} = {\lim_{m\rightarrow\infty}} \; \E { \frac{1}{m} {\sum\limits_{k = 0}^{m-1} } \,\tilde{c}(s_k,T(s_k)) }.
\end{equation}
In contrast, a stationary randomized policy (SRP) is a sequence of maps $\varphi = \{\psi,\psi,\ldots \}$ such that for a state $s_k$, 
$\psi(s_k,\cdot)$ is a probability distribution over the set of feasible actions in state $s_k$. Such a policy does not change with time.
The single-stage cost $\tilde{d}(s_k)$ of an SRP $\varphi$ is given by
\begin{equation}
\label{eqn:randomized-cost-fn1}
\tilde{d}(s_k) = \sum_{a \in A(s_k)} \psi(s_k,a)\,\tilde{c}(s_k,a),
\end{equation}
where $a$ gives the energy split in state $s_k$. The long-run average cost of an SRP $\varphi$ is 
\begin{equation}
\label{eqn:avgcost-srp1}
\tilde{\lambda}^{\varphi} = {\lim_{m\rightarrow\infty}} \frac{1}{m} {\sum\limits_{k = 0}^{m-1} } \tilde{d}(s_k).
\end{equation}

We observe that the term $q^i_{k}$ in (\ref{eqn:cost-structure-not-used}) does not include the effect of action explicitly.
Hence we modify the cost function to include the effect of the action taken explicitly into the cost function. 
In order to enable reformulation of the average cost objective in the modified form, we prove the following lemma. Define
\begin{equation}
\label{eqn:actual-avg-cost}
 \lambda^{\pi} = {\lim_{m\rightarrow\infty}} \E {\frac{1}{m} {\sum\limits_{k = 0}^{m-1} } 
		  \,{\sum_{i = 1}^{n}} \left(q^i_{k}-g(T^i(s_{\!k})) \right)^{+} }.
\end{equation}
\begin{lemma}
\label{lemma1}
 Let $q^i_{k}, \: 1 \leq i \leq n$, $T^i(s_{\!k}), \: 1 \leq i \leq n$ and $g$ be as before and let $\E {X^i}, \: 1 \leq i \leq n$ denote the mean of the i.i.d random variables $X^i, \: 1 \leq i \leq n$. Then
\begin{displaymath}
  \lambda^{\pi} = \tilde{\lambda}^{\pi} - \sum\limits_{i=1}^{n} \E {X^i}
\end{displaymath}
for all policies $\pi$.
\end{lemma}
\begin{proof} 
Using state evolution equations (\ref{eqn:qk-evolution})-(\ref{eqn:ek-evolution}),
  \begin{displaymath}
   \lim_{m \rightarrow \infty} E \left[ \frac{1}{m} \sum\limits_{k=0}^{m-1} \sum\limits_{i=1}^{n} \left( q^i_k - g(T^i(s_{\!k}))\right)^{+}\right]
  \end{displaymath}
 \begin{align*}
&= \lim_{m \rightarrow \infty} E \left[ \frac{1}{m} \sum\limits_{k=0}^{m-1} \sum\limits_{i=1}^{n} \left( q^i_{k+1} - X^i_{\!{k}}\right) \right] \\  
&= \lim_{m \rightarrow \infty} E \left[ \sum\limits_{i=1}^{n} \frac{1}{m} \sum\limits_{k=0}^{m-1} \left( q^i_{k+1} - X^i_{\!{k}} \right) \right] \\ %
&= \sum\limits_{i=1}^{n} \lim_{m \rightarrow \infty} E \left[ \frac{1}{m} \sum\limits_{k=0}^{m-1} \left( q^i_{k+1} - X^i_{\!{k}} \right) \right] \\
 &= \sum\limits_{i=1}^{n} \left\{ \lim_{m \rightarrow \infty} 
   E \left[ \frac{1}{m} \sum\limits_{k=0}^{m-1} q^i_{k+1} \right] - 
   \lim_{m \rightarrow \infty} 
   E \left[ \frac{1}{m} \sum\limits_{k=0}^{m-1} X^i_{k} \right] \right\} \\
&= \sum\limits_{i=1}^{n} \left\{ \lim_{m \rightarrow \infty} 
   E \left[ \frac{1}{m} \left( \sum\limits_{k=0}^{m-1} q^i_{k} + q^i_m - q^i_0 \right) \right] - 
   \lim_{m \rightarrow \infty} E \left[ \frac{1}{m} \sum\limits_{k=0}^{m-1} X^i_{k} \right] \right\} \\
&= \sum\limits_{i=1}^{n} \left\{ \lim_{m \rightarrow \infty} 
   E \left[ \frac{1}{m} \left( \sum\limits_{k=0}^{m-1} q^i_{k} + q^i_m - q^i_0 \right) \right] - 
   E \left[ X^i \right] \right\} \\
&= \sum\limits_{i=1}^{n} \left\{  \lim_{m \rightarrow \infty}
   E \left[ \frac{1}{m} \sum\limits_{k=0}^{m-1} q^i_{k} \right] +
   \lim_{m \rightarrow \infty} E \left[ \frac{1}{m} \left( q^i_m - q^i_0 \right) \right] - 
   E \left[ X^i \right] \right\} \\
%
&= \sum\limits_{i=1}^{n} \left\{ \lim_{m \rightarrow \infty} E \left[ \frac{1}{m} \sum\limits_{k=0}^{m-1} q^i_{k} \right] \right\} - \sum\limits_{i=1}^{n} E \left[ X^i \right] \\
&= \tilde{\lambda}^{\pi} - \sum\limits_{i=1}^{n} E \left[ X^i \right].
  \end{align*}
The second last equality above follows from the fact that $\lim\limits_{m \rightarrow \infty} E \left[ \frac{1}{m} \left( q^i_m - q^i_0 \right) \right] = 0$.
The claim follows. 

\end{proof}
The linear relationship between $\tilde{\lambda}^{\pi}$ and $\lambda^{\pi}$ enables us to
define the new single-stage cost function as:
\begin{equation}
\label{eqn:cost-fn}
c(s_k,T_k) = {\sum_{i = 1}^{n}} (q^i_{k} - g(T^i(s_{\!k})))^{+}.
\end{equation}
With this single-stage cost function, the long-run average cost of an SDP $\pi$ is given by
\begin{equation}
\label{eqn:avgcost}
\lambda^{\pi} = {\lim_{m\rightarrow\infty}} \; \E { \frac{1}{m} {\sum\limits_{k = 0}^{m-1} } c(s_k,T(s_k)) }.
\end{equation}
The single-stage cost $d(s_k)$ of an SRP $\varphi$ is given by
\begin{equation}
\label{eqn:randomized-cost-fn}
d(s_k) = \sum_{a \in A(s_k)} \psi(s_k,a)\,c(s_k,a),
\end{equation}
where $a$ gives the energy split. The long-run average cost of an SRP $\varphi$ is 
\begin{equation}
\label{eqn:avgcost-srp}
\lambda^{\varphi} = {\lim_{m\rightarrow\infty}} \; \frac{1}{m} {\sum\limits_{k = 0}^{m-1} } d(s_k).
\end{equation}
It can be inferred from Lemma \ref{lemma1} that a policy which minimizes the average cost in \eqref{eqn:cost-fn} (or \eqref{eqn:avgcost-srp}) 
will also minimize the average cost given by \eqref{eqn:avgcost1} (or \eqref{eqn:avgcost-srp1}).
In this paper we are interested in finding stationary policies (deterministic or randomized) which optimally share the energy among a set of nodes.
Therefore our aim is to find policies which minimize the average cost per step, when the single-stage cost is given by (\ref{eqn:cost-fn}).

Any stationary optimal policy minimizes the average cost of the system over all policies.
Let $\pi^*$ be an optimal policy and $\Pi$ be the set of all policies. The average cost of policy $\pi^*$ is denoted $\lambda^*$.
Then \begin{displaymath}
      \lambda^* = \inf_{\pi \in \Pi} \lambda^{\pi}.
     \end{displaymath}
The policy corresponding to the above average cost minimizes the sum of (data) queue lengths of all nodes.
By Little's law, under stationarity, the average sum of data queue lengths at the sensor nodes is proportional to the
average waiting time or delay of the arrivals (bits). 
Hence an average cost optimal policy minimizes the stationary mean delay as well.

The class of stationary deterministic policies is contained in the class of stationary randomized policies and
in the system we consider, an optimal policy is known to exist in the class of stationary deterministic policies.
We provide an algorithm which finds an optimal SDP.
The algorithm is computationally efficient for small state and action spaces. 
However for large state-action spaces, the algorithm computations are expensive. 
To mitigate this problem, we provide approximation algorithms which find near-optimal stationary policies for the system. These algorithms
are described in the following sections.
%

\section{Energy Sharing Algorithms}
\label{sec:2nodes-algo}
\subsection{Background}
Consider an optimal SDP ${\pi}^*$ for the energy sharing MDP. Then $\lambda^*$ corresponds to the average cost of the policy $\pi^*$.
Suppose $i_r$ is a reference state in the MDP. For any state $i \in S$, let $h^*(i)$ be the relative (or the differential) cost 
defined as the minimum of the difference between the expected cost to reach state $i_r$ from $i$ and the expected cost incurred 
if the cost per stage was $\lambda^{*}$. The quantities $\lambda^{*}$ and $h^{*}(i) , i \in S$ satisfy the Bellman Equation:
\begin{equation}
\label{eqn:bellman-eqn}
 \lambda^* + h^*(i) = {\min_{a \in A(i)}} \left( c(i,a) +  {\sum_{j \in S}} \; p(i,a,j)\,h^*(j)\right),
\end{equation}
where $p(i,a,j)$ is the probability that the system will move from state $i$ to state $j$ under action $a$.
We denote by $Q^*(i,a)$, the optimal differential cost of any feasible state-action tuple $(i,a)$ as follows:
\begin{equation}
\label{eqn:Q*}
 Q^*(i,a) = c(i,a) + \sum_{j \in S} p(i,a,j)h^{*}(j).
\end{equation}
Equation (\ref{eqn:bellman-eqn}) can now be rewritten as
\begin{equation}
\label{eqn:vMinOverq}
\lambda^* + h^*(i) = {\min_{a \in A(i)}} Q^*(i,a), \quad \forall i \in S
\end{equation}
or alternately
\begin{equation}
\label{eqn:vMinOverq-lambda*}
 h^*(i) = {\min_{a \in A(i)}} Q^*(i,a) - \lambda^*, \quad \forall i \in S.
\end{equation}
Plugging \eqref{eqn:vMinOverq-lambda*} into \eqref{eqn:Q*}, one obtains
\begin{equation}
 Q^*(i,a) =  c(i,a) + {\sum_{j \in S}} \; p(i,a,j)\left[ {\min_{b \in A(j)}} Q^{*}(j,b) - \lambda^* \right]
\end{equation}
or
\begin{equation}
 \label{eqn:q-bellmanOpt}
\lambda^* + Q^*(i,a) =  c(i,a) + {\sum_{j \in S}} \; p(i,a,j)  {\min_{b \in A(j)}} Q^{*}(j,b),  \quad \forall i \in S, \forall a \in A(i).
\end{equation}
Equation \eqref{eqn:q-bellmanOpt} is also referred to as the Q-Bellman equation. The important thing to note is that whereas the Bellman equation \eqref{eqn:bellman-eqn}
is not directly amenable to stochastic approximation, the Q-Bellman equation \eqref{eqn:q-bellmanOpt} is; because of the fact that the minimization operation
in \eqref{eqn:q-bellmanOpt} is inside the conditional expectation unlike \eqref{eqn:bellman-eqn} (where it is outside of it).
If the transition probabilities and the cost structure of the system model are known, 
then (\ref{eqn:q-bellmanOpt}) can be solved using dynamic programming techniques \cite{tsitsiklis1999average}.
When the system model is not known (as in the problem we study), the Q-learning algorithm can be used to obtain optimal policies. 
This learning algorithm solves (\ref{eqn:q-bellmanOpt}) in an online manner using simulation to obtain an optimal policy. 
It is described in the following subsection.
\subsection{Relative Value Iteration based Q-Learning}
\label{sec:ql}
Q-learning is a stochastic iterative, simulation-based algorithm that aims to find the $Q^{*}(i,a)$ values for all feasible state-action pairs 
$(i,a)$. It is a model-free learning algorithm and proceeds by assuming that the transition probabilities $p(i,a,j)$ are unknown. 
Initially Q-values for all state-action pairs are set to zero, i.e., $Q_{0}(i,a) = 0, \forall i \in S, \; a \in A(i)$. 
Then $\forall k \geq 0$, the Q-learning update \cite{abounadi2001learning} for a state-action pair visited during simulation is carried out as follows:
\begin{equation}
\label{eqn:q-learning-update}
Q_{k+1}(i,a) = (1-\alpha(k))Q_{k}(i,a) +  \alpha(k) \left(c(i,a) + {\min_{b \in A(j)}} Q_{k}(j,b) - {\min_{u \in A(i_r)}} Q_{k}(i_r,u)\right),
\end{equation}
where $i$ is the current state at decision time $k$ and $i_r$ is the reference state. 
The action in state $i$ is selected using one of the exploration mechanisms described below. State $j$ corresponds to the `next' state that is
obtained from simulation when the action $a$ is selected in state $i$.
Also, $\alpha(k), \; k \ge 0$ is a given step-size sequence such that $\alpha(k) > 0 , \forall k \geq 0$ and satisfies the following conditions:
\[
\text{ } {\sum_{k}} \alpha(k) = \infty \text{ and } {\sum_{k}} \alpha^2(k) < \infty.
\]
Let $t(k) = \sum\limits_{i=0}^{k-1} \alpha(i)$, $k \geq 1$, with $t(0) = 0$. Then, $t(k)$, $k \geq 0$ corresponds
to the ``timescale'' of the algorithm's updates. The first condition above ensures that $t(k) \rightarrow \infty$
as $k \rightarrow \infty$.
This ensures that the algorithm does not converge prematurely. The second condition makes sure that the noise asymptotically vanishes.
These conditions on step sizes guarantee the convergence of Q-learning to the optimal state-action value function, 
see \cite{abounadi2001learning} for a proof of convergence of the algorithm.
The update (\ref{eqn:q-learning-update}) is carried out for the state-action pairs visited during simulation.
The exploration mechanisms we employ are as follows:
\begin{enumerate}
 \item $\epsilon$-greedy: In the energy sharing problem, the number of actions feasible in every state is finite. Hence there exists an action $a_m$
for state $i$ such that $Q_{k}(i,a_m) \leq Q_{k}(i,a^{'}), \; \forall a^{'} \in A(i), \; \forall k \geq 0$. We choose $\epsilon \in (0,1)$. 
In state $i$, action $a_m$ is picked with probability $1-\epsilon$, while any other action is picked with probability $\epsilon$.
 \item UCB Exploration: Let $N_i(k)$ be the number of times state $i$ is visited until time $k$. Similarly let $N_{i,a}(k)$ be the number of times action
$a$ is picked in state $i$ upto time $k$. The Q-value of state-action pair $(i,a)$ at time $k$ is $Q_{k}(i,a)$. 
When the state $i$ is encountered at time $k$, the action for this state is picked according to the following rule:
\begin{equation}
 \label{eqn:uct}
a' = \amax {a \in A(i)} \left( -Q_{k}(i,a) + \beta \, \sqrt{ \frac{\ln N_i(k)}{N_{i,a}(k)} }\right),
\end{equation}
where $\beta$ is a constant. The first term on the right hand side gives preference to an action that has yielded good performance in the past visits to state
$i$, while the second term gives preference to actions that have not been tried out many times so far, relative to $\ln N_i(k)$.
\end{enumerate}
\begin{remark}
The convergence rates for the discounted Q-learning have been studied in
 \cite{tsitsiklis1994asynchronous,kearns1999finite,even2004learning}.
 The finite-time bounds to reach an $\epsilon$-optimal policy by following the Q-learning rule are given in  
\cite{tsitsiklis1994asynchronous,kearns1999finite,even2004learning}. 
In the Q-learning algorithm, to explore the value of different states and actions, 
one needs to visit each state-action pair infinitely often. 
However, in practice, depending on the size of the state-action space, we need to simulate the Q-learning algorithm 
so that each state-action pair is visited a sufficient number of times. In our experiments for the case of two sensor nodes, the 
size of the state space is of the order of $10^5$ and we ran our algorithm for $10^8$ iterations.
\end{remark}

Once we determine $Q^{∗}(i,a)$ for all state-action pairs, we
can obtain the optimal action for a state $i$ by choosing the action that minimizes $Q^{*}(i,a)$. So
\begin{equation}
 a^{*} = \amin {a \in A(i)} \,Q^{*}(i,a).
\end{equation}
It should be noted that the Q-learning algorithm does not need knowledge of the cost structure and transition probabilities, and
it learns an optimal policy by interacting with the system.
\section{Approximation Algorithms}
\label{sec:approximation}
The learning algorithm described in Section \ref{sec:2nodes-algo} is an iterative stochastic algorithm that learns the optimal energy split. 
This method requires that the $Q(s,a)$ values be stored for all $(s,a)$ tuples. The values of $Q(s,a)$ for each $(s,a)$ tuple are updated in 
\eqref{eqn:q-learning-update} over a number of iterations using adequate exploration.
These updations play a key role in finding the optimal control for a given state. 
Nevertheless for large state-action spaces these computations are expensive as every lookup operation and updation require memory access. 
For example, if there are two nodes sharing energy and buffer sizes are $E_{\!_{MAX}} = D_{\!_{MAX}} = 30$, 
then the number of $(s,a)$ tuples would be of the order $10^6$, 
which demands enormous amount of computation time and memory space. 
This condition is exacerbated when the number of nodes that share energy increases. 
For instance, in the case of four nodes sharing energy with $E_{\!_{MAX}} = D_{\!_{MAX}} = 30$, we have $|S \times A| \approx 30^9$. 
Thus, we have a scenario where the state-action space can be extremely large.

To mitigate this problem, we propose two algorithms that are both based on certain threshold features.
Both algorithms tackle the curse of dimensionality, by reducing the computational complexity. 
We describe below our threshold based features, following which we describe our algorithms.

\subsection{Threshold based Features}
\label{sec:t-features}
The fundamental idea of threshold based features is to cluster states in a particular manner, based on the properties of the differential value functions.
The following proposition proves the monotonicity property of the differential value functions for the scenario where there is a single node and an EH source.
This simple scenario is considered for the sake of clarity in the proof.
\begin{proposition}
\label{lemma2}
Let $H^{*}(q,E)$ be the differential value of state $(q,E)$. Let $q < q^L \leq D_{\!{MAX}}$ and $E_{\!{MAX}} \geq E^L > E$, respectively. Then,
\begin{align}
 H^{*}(q,E) &\leq H^{*}(q^{L},E), \\
 H^{*}(q,E) &\geq H^{*}(q,E^{L}).
\end{align}
\end{proposition}
\begin{proof}
Let $J(s)$ be the total cost incurred when starting from state $s$. Define the Bellman operator $L: \R^{n} \rightarrow \R^{n}$ as
\begin{displaymath}
 (L\,J)(s) =  \min_{T \in A(s)}(c(s,T) + \E {J(s')}),
\end{displaymath}
where $s'$ corresponds to the next state after $s$ and $T$ corresponds to the action taken in state $s$. As noted in Section \ref{sec:t-features},
we show the proof for a single node and EH source. The proof can be easily generalized to multiple nodes.
Thus the state $s$ corresponds to the tuple $(q,E)$. Hence the above equation can be rewritten as
\begin{displaymath}
 (L\,J)(q,E) =  \min_{T \in A(s)} \left(c(q,E,T) + \E {J(q^{'},E^{'})} \right),  \quad \forall (q,E) \in S.
\end{displaymath}
We consider the application of the operator $L$ on the differential cost function $H(\cdot)$.
We set out to prove this proposition using the relative value iteration scheme (see \cite{bertsekas1996neuro}).
For this, we set a reference state $r \triangleq (q_r,E_r) \in S$.
The cost function in our case is $(q-g(T))^{+}$.
Initially the differential value function has value zero for all states $(q,E) \in S$, i.e., $H(q,E) = 0$, $\forall (q,E) \in S$.
Then for some arbitrary $(q,E) \in S$ we have
\begin{align*}
 L\,H(q,E) &= \min_{T \in A(q,E)} \left( (q-g(T))^{+} + \E {H(q^{'},E^{'})} \right) - L \,H(q_r,E_r) \\
             &= \min_{T \in A(q,E)} ( (q-g(T))^{+} )- L \,H(q_r,E_r)
\end{align*}
since $H(q^{'},E^{'}) = 0$, $\forall (q^{'},E^{'}) \in S$. Let $T_m$ be the value of $T$ achieving the minimum in the first term of RHS. Then
\begin{displaymath}
L\,H(q,E) = (q-g(T_{m}))^{+} - L \,H(q_r,E_r).
\end{displaymath}
Now consider the differential value of state $q^{L}$ where $ q^{L} > q$. Thus, consider
\begin{align*}
 L\,H(q^{L},E) &= \min_{T \in A(q^{L},E)} \left( (q^{L}-g(T))^{+} + \E {H(q^{'},E^{'})} \right) - L \,H(q_r,E_r) \\
             &= \min_{T \in A(q^{L},E)} ( (q^{L}-g(T))^{+} )- L \,H(q_r,E_r) \\
	     &= (q^{L}-g(T_{\!L}))^{+} - L \,H(q_r,E_r),
\end{align*}
where $T_{\!L}$ is the value of $T$ for which the minimum of the expression $(q^{L}-g(T))^{+}$, in the above equations, is achieved.
We have
\begin{align}
\label{eqn:1}
 L\,H(q,E) &= (q-g(T_{m})) - L \,H(q_r,E_r) \nonumber\\
           &\leq (q-g(T_{\!L})) - L \,H(q_r,E_r) \nonumber\\
	   &\leq (q^{L}-g(T_{\!L})) - L \,H(q_r,E_r) \nonumber\\
	   &= L\,H(q^{L},E).
\end{align}
We have $H(q^{L},E) \geq H(q,E)$ since these values are initialized to zero and from \eqref{eqn:1}, $L\,H(q^{L},E) \geq L\,H(q,E)$.
Now consider the differential value function of the state $(q,E^{L})$ where $E^{L} > E$.
\begin{align*}
 L\,H(q,E^{L}) &= \min_{T \in A(q,E^{\!L})} \left( (q-g(T))^{+} + \E {H(q^{'},E^{'})} \right) - L \,H(q_r,E_r) \\
	       &=  \min_{T \in A(q,E^{\!L})} \left( (q-g(T))^{+}\right) - L \,H(q_r,E_r) \\
	       &=  (q-g(T_{\!E}))^{+} - L \,H(q_r,E_r),
\end{align*}
where $T_{\!E}$ is the value of $T$ for which the minimum of the expression $(q - g(T))^{+}$, in the above
equations, is achieved. We have,
\begin{align}
\label{eqn:2}
 L\,H(q,E^{L}) &=  (q-g(T_{\!E}))^{+} - L \,H(q_r,E_r) \nonumber \\
	       &\leq (q-g(T_{\!m}))^{+} - L \,H(q_r,E_r) \nonumber \\
	       &\leq L\,H(q,E).
\end{align}
Since $H(q,E^{L})$, $H(q,E)$ are initialized to zero, we have $H(q,E^{L}) \leq H(q,E)$ and from \eqref{eqn:2}, $L\,H(q,E^{L}) \leq L\,H(q,E)$.
We prove the following statements using mathematical induction:
\begin{align*}
 L^{k}\,H(q,E) &\leq L^{k}\,H(q^{L},E) \qquad \forall k \geq 0,\\
 L^{k}\,H(q,E) &\geq L^{k}\,H(q,E^{L}) \qquad \forall k \geq 0.
\end{align*}
We have seen above that the two statements are true for both $k = 0$ and $k = 1$, respectively.
Lets consider the first statement and assume that the statement holds for some $k$. We then prove that it holds for $(k+1)$. Consider 
\[ L^{k+1}H(q,E) = \underset{T \in A(q,E)}{\min} \left((q - g(T))^{+} + \E {L^{k}H(q^{'},E^{'})} \right) - L^{k}\,H(q_r,E_r). \]
Assume $T_m$ is the value of $T$ at which the minimum of $((q - g(T))^{+} + \E {L^{k}H(q^{'},E^{'})} )$
is attained. Then,
\[ L^{k+1}H(q,E) = ((q-g(T_m))^{+} +  \E {L^{k}H(q - g(T_m) + x, E - T_m + y)} ) - L^{k}\,H(q_r,E_r), \]
where $x,y$ are obtained from independent random distributions. Similarly, we get
\[ L^{k+1}H(q^{L},E) = ((q^{L}-g(T_L))^{+} + \E { L^{k}H(q^{L} - g(T_{\!L}) + x, E - T_L + y)} ) - L^{k}\,H(q_r,E_r),\]
where $T_{\!L}$ is the value of $T$ for which the minimum in the expression $((q^{L}-g(T))^{+} + \E {L^{k}H(q^{'},E^{'})} )$ is achieved.
\begin{align*}
L^{k+1}H(q,E) &= ((q-g(T_m))^{+} + \E {L^{k}H(q - g(T_m) + x, E - T_m + y)} ) - L^{k}\,H(q_r,E_r)\\
	      &\leq ((q - g(T_{\!L}))^{+} + \E {L^{k}H(q - g(T_{\!L}) + x, E - T_{\!L} + y)} ) - L^{k}\,H(q_r,E_r) \\
	      &\leq ((q^{L} - g(T_{\!L}))^{+} +  \E {L^{k}H(q^{L} - g(T_{\!L}) + x, E - T_{\!L} + y)} )- L^{k}\,H(q_r,E_r),
\end{align*}
since the property holds true for $L^{k}H$, i.e., $L^{k}H(q,E) \leq L^{k}H(q^{L},E)$. Thus,
\begin{align*}
L^{k+1}H(q,E) &\leq ((q^{L} - g(T_{\!L}))^{+} + \E{ L^{k}H(q^{L} - g(T_{\!L}) + x, E - T_{\!L} + y)} ) - L^{k}\,H(q_r,E_r)\\
	      &=     L^{k+1}H(q^{L},E).
\end{align*}
Hence,
\begin{equation}\label{induction1}
L^{k}H(q,E) \leq L^{k}H(q^{L},E) \quad\forall k \geq 0.
\end{equation}
Similarly we get,
\begin{align*}
L^{k+1}H(q,E^{L}) &= ((q - g(T_E))^{+} + \E{ L^{k}H(q - g(T_E) + x, E^{L} - T_E + y)} ) - L^{k}\,H(q_r,E_r)\\
		  &\leq ((q - g(T_m))^{+} +  \E{L^{k}H(q - g(T_m) + x, E^{L} - T_m + y)} ) - L^{k}\,H(q_r,E_r)\\
		  &\leq ((q - g(T_m))^{+} + \E{ L^{k}H(q - g(T_m) + x, E - T_m + y)} ) - L^{k}\,H(q_r,E_r)\\
		  &= L^{k+1}H(q,E),
\end{align*}
hence by mathematical induction on $k$ we get,
\begin{equation}\label{induction2}
L^{k}H(q,E^{L}) \leq L^{k}H(q,E) \quad\forall k \geq 0.
\end{equation}
As a consequence of the relative value iteration scheme (\cite{puterman1994markov}), 
when $k \rightarrow \infty$, $L^{k}H \rightarrow H^{*}$ with $H^{*}(q_r,E_r) = \lambda^{*}$. Thus, from 
\eqref{induction1} and \eqref{induction2} as $k \rightarrow \infty$, we obtain
\[H^{*}(q,E) \leq H^{*}(q^{L}, E) \]
\[H^{*}(q,E) \geq H^{*}(q, E^{L}) .\]
The claim now follows.
\end{proof}
\par Proposition \ref{lemma2} can be easily generalized to multiple nodes in the following manner.
Suppose there are $n$ nodes and one EH source. Let $s = (q^1,\ldots,q^j,\ldots,q^n,E)$
and $s' = (q^1,\ldots,q^j_{L},\ldots,q^n,E)$, where $q^j_{L} > q^j$.
The states $s$ and $s'$ differ only in the data buffer queue lengths of node $j$, 
while the data buffer queue lengths of other nodes remain the same and so does the
energy buffer level. Then it can be observed that \( H^{*}(q^1,\ldots,q^j,\ldots,q^n,E) \leq H^{*}(q^1,\ldots,q^j_{L},\ldots,q^n,E) \).
In a similar manner, let state $s'' = (q^1,q^2,\ldots,q^n,E^L) $ and $E^L > E$. Then states $s \text{ and } s''$ differ only in the
energy buffer levels. Hence \( H^{*}(q^1,\ldots,q^n,E) \) \(\geq H^{*}(q^1,\ldots,q^n,E^L) \).
This proposition provides us a method which is useful for clustering states. 
\begin{remark}
The monotonicity property of the differential value function $H^*$ provides a justification to group nearby states 
to form an aggregate state. The value function of the aggregated state will be the average of the value function of the 
states in a partition.
If the difference between values of states in a cluster is not much, the value function of aggregated state will 
be close to the value function of the unaggregated state.
Thus, the policy obtained from the aggregated value function is likely to be close to the policy obtained from unaggregated states. 
Without the monotonicity property, states may be grouped arbitrarily and consequently, state aggregation may not 
yield a good policy. 
\end{remark}
\begin{remark}
 In the case of MDP with large state-action space, one goes for function approximation based methods 
(see Chapter 8 in \cite{sutton1998introduction}). However, if one combines Q-learning with function approximation, 
we do not have convergence guarantees to the optimal policy unlike Q-learning without function approximation 
(Q-learning with tabular representation \cite{sutton1998introduction}). However, when 
Q-learning is combined with state-aggregation (QL-SA) we continue to have convergence guarantees 
(see Section 6.7 in \cite{bertsekas1996neuro}). Q-learning using state aggregation 
can produce good policies only when the value function has a monotonicity structure, which is proved in the previous remark.
\end{remark}

\subsubsection{Clustering}
The data and energy buffers are quantized and using this we formulate the aggregate state-action space.
The quantization of buffer space is described next.
We predefine data buffer and energy buffer partitions (or quantization levels) $d_1,d_2,\ldots,d_s$ and $e_1,e_2,\ldots,e_r$ respectively. 
The partition (or quantization level) $d_i, (i \in \{1,\ldots, s\})$ corresponds to a given range
 $(x^i_{\!_L},x^i_{\!_U})$ and is fixed, 
where $x^i_{\!_L}$ and $x^i_{\!_U}$ represent the prescribed lower and upper data buffer level limits. 
In a similar manner the quantization level $e_j,(j \in \{1,\ldots, r\})$ (or energy buffer partition) 
corresponds to a given interval $(y^i_{\!_L},y^i_{\!_U})$,
where $y^i_{\!_L}$ and $y^i_{\!_U}$ represent the prescribed lower and upper energy buffer level limits. 
As an example, suppose $D_{\!_{MAX}} = E_{\!_{MAX}} = 10$ and each of the buffers are quantized into three levels, 
i.e., $s = r = 3$. An instance of data and energy buffer partition ranges in this scenario can be 
$y^1_{\!_L} = x^1_{\!_L} = 0, y^1_{\!_U} = x^1_{\!_U} = 3, y^2_{\!_L} =  x^2_{\!_L} = 4,
x^2_{\!_U} = y^2_{\!_U} = 7,x^3_{\!_L} = y^3_{\!_L} = 8, y^3_{\!_U} = x^3_{\!_U} = 10$.
Here Partition $1$ corresponds to the number of data (energy) bits (units) in the range $(0,3)$, while
Partition $3$ corresponds to the number of data (energy) bits (units) in the range $(8,10)$.
The following inequalities hold with respect to the partition limits:
\begin{align*}
 0 = x^1_{\!_L} < x^1_{\!_U} &< x^2_{\!_L} < x^2_{\!_U} < \ldots < x^s_{\!_L} < x^s_{\!_U} = D_{\!_{MAX}} \text{   and}\\
&x^i_{\!_U} + 1 =  x^{i+1}_{\!_L},	\quad 1 \leq i \leq s-1 . 
\end{align*}
Similarly,
\begin{align*}
 0 = y^1_{\!_L} < y^1_{\!_U} &< y^2_{\!_L} < y^2_{\!_U} < \ldots < y^s_{\!_L} < y^s_{\!_U} = E_{\!_{MAX}} \text{   and}\\
&y^i_{\!_U} + 1 =  y^{i+1}_{\!_L},	\quad 1 \leq i \leq r-1 .
\end{align*}

\subsubsection{Aggregate States and Actions}
We define an aggregate state as $s' = \{l^{^1},\ldots,l^{^{n+1}}\}$, where for
$1 \leq i \leq n$, $l^{^i}$ is the data buffer level for the $i^{th}$ node and $l^{n+1}$ is the energy buffer level.
So $l^{^i} \in \{1,\ldots, s\},\; 1 \leq i \leq n$ and $l^{^{n+1}} \in \{1,\ldots, r\}$.
An aggregate action corresponding to the state $s'$ is an $n$-tuple $t^{'}$ of the form $t^{'} = (t^1,\ldots,t^n)$, 
where $t^i \in \{1,\ldots, l^{^{n+1}}\} ,\; 1 \leq i \leq n$.
Each component in $t^{'}$ indicates an energy level. By considering the data level in all the nodes, 
the controller decides on an energy level for each node.
Thus the energy level indicates the
energy partition which can be supplied to the node. For instance, if $D_{\!_{MAX}} = E_{\!_{MAX}} = 15$, 
$s = r = 3$ and there are two nodes in the system, then an example aggregate state is $s' = (1,1,3)$. Suppose the controller selects
the aggregate action $t^{'} = (2,1)$, which means that the controller decides to give $u$ number of energy bits to Node $1$, 
and $v$ number of energy bits to Node $2$, with $y^2_{\!_L} \leq u \leq y^2_{\!_U}$ and $y^1_{\!_L} \leq v \leq y^1_{\!_U}$, respectively.

\subsubsection{Cardinality Reduction}
Note that $s \ll D_{\!_{MAX}}$, $r \ll E_{\!_{MAX}}$.
Let the aggregated state and action spaces be denoted by $S^{'} \text{ and }A^{'}$ respectively. The aggregated state-action space has
cardinality $|S^{'} \times A^{'}|$.
Thus, the cardinality of the state-action space is reduced to a great extent by aggregation. 
For instance, in the case of four nodes sharing energy from one EH source and $E_{\!_{MAX}} = D_{\!_{MAX}} = 30$, the cardinality of the state-action
space without state-aggregation is $|S \times A| \approx 30^9$. However, with four partitions each for the data and energy buffers, the 
cardinality of the state-action space with aggregation is $|S^{'} \times A^{'}| \approx 4^9$.

\subsection{Approximate Learning Algorithm}
\label{sec:approx-alg}
We now explain our approximate learning algorithm for the energy sharing problem. It is based on Q-learning and state aggregation. 
Although the straightforward Q-learning algorithm described in Section \ref{sec:2nodes-algo} requires complete state information and is not computationally 
efficient with respect to large state-action spaces, its state-aggregation based counterpart requires significantly less computation and memory space. 
Also our experiments show that we do not compromise much on the policy obtained either (see Fig. \ref{plot:2n-ql-qlsa-gr-cnql}).

\subsubsection{Method}
Let $s' = \{l^{^1}_k,\ldots,l^{^{n+1}}_k\}$ be the aggregate state at decision time $k$. 
The action taken in $s'$ is $t' = (t^1_k,\ldots,t^n_k)$. 
The Q-value $Q(s',t{'})$ indicates how good an aggregate state-action tuple is. 
The algorithm proceeds with the following update rule:
\begin{equation}
\label{eqn:ql-sa}
 Q_{k+1}(s',t') = (1-\alpha(k))Q_{k}(s',t') + \alpha(k)\left( c(s',t') + 
 \underset{b \in A^{'}(j')} \min Q_k(j',b)  - \underset{u \in A^{'}(r')} \min Q_{k}(r',u)\right),
\end{equation}
where $j'$ is the aggregate state obtained by simulating action $t'$ in state $s'$. 
Also, $r'$ is a reference state and $\alpha(k), \; k \ge 0$ is a positive step-size schedule satisfying the conditions mentioned in Section \ref{sec:ql}.
To facilitate exploration, we employ the mechanisms described in Section \ref{sec:ql}.
Convergence of Q-learning with state aggregation is discussed in Section 6.7 of \cite{bertsekas1996neuro}.

\begin{remark}
The aggregate state in every step of the iteration \eqref{eqn:ql-sa} is computed by 
knowing the amount of data present in each sensor node. 
A viable implementation would just need a mapping of the buffer levels to these partitions,
using which the controller can compute the aggregate state for any combination
of buffer levels. Since this method requires storing of Q-value of 
the aggregate state-action pair and $|S^{'} \times A^{'}| \ll |S \times A|$, the number of Q-values stored is much
less compared to the unaggregated Q-learning algorithm. The computational complexity
of the method described above is dependent on the size of the aggregate state-action space and 
the number of iterations required to converge to an optimal policy (w.r.t the aggregate state-action space).
For instance, in the case of four sensor nodes, the size of the state-action space grows to $\approx 30^9$ with the 
data and energy buffer sizes being $30$ each.
The number of iterations that the above method requires to find a near-optimal policy is $10^{9}$ with six partitions of the buffer size
as compared to Q-learning without state aggregation (Section \ref{sec:ql}) which requires at least $10^{11}$ iterations.
\end{remark}
\begin{remark}
It must be observed that using \eqref{eqn:ql-sa}, the controller decides the partition and not the number of energy bits to be distributed, i.e., it finds an optimal aggregate action for every aggregate state. 
It follows from this that, in order to find the aggregate action for an aggregate state, 
the knowledge of the exact buffer levels is not required 
(since this is based on the Q-values of aggregate state-action pairs).
In this manner \eqref{eqn:ql-sa} is beneficial. 
The optimal policy obtained using (\ref{eqn:ql-sa}) would indicate only the energy levels.
An added advantage of the above approximation algorithm is 
that the cost structure discussed in Section \ref{sec:es-prob-as-an-mdp} holds good
here as well.
\end{remark}
\subsubsection{Energy distribution}
\label{subsubsec:en-distr}
Note that once an aggregate action is chosen for a state, the energy division is random adhering to the action levels chosen. 
For instance, lets assume that there are two sensor nodes in the system.
Data and energy buffers have three partitions each and thus \( s = 3, \: r = 3\). 
Here $y^1_{\!_L} = 0$ and $y^3_{\!_U} = E_{\!_{MAX}}$.
Suppose the number of energy bits in the energy buffer is $z$ and those bits belong to partition $3$. 
Let the number of data bits at nodes $1 \text{ and }2$ be $x$ and $y$, respectively. Here $x \text{ and } y$ belong to partition $2$.
Hence the aggregate state is $(2,2,3)$. The controller decides on the aggregate action $(1,2)$.
Thus $x^1_{\!_L}$ bits of energy is provided to Node 1, while Node 2 is given $x^2_{\!_L}$ bits of energy.
The remaining number of bits in the buffer will be $r = z - (x^1_{\!_L} + x^2_{\!_L})$. In order to distribute these bits, the proportions of data  
$p_1 = \frac{x}{x+y}$ and $1 - p_1 = \frac{y}{x+y}$ are computed. Each of the $r$ bits are supplied to Node $1$ with probability $p_1$ and to Node $2$
with probability $1-p_1$. If $u \text{ and }v$ represent the total number of energy bits provided to Nodes 1 and 2 respectively, then 
$u \leq x^1_{\!_U}$, $v \leq x^2_{\!_U}$ and $(u - x^1_{\!_L}) + (v - x^2_{\!_L}) \leq r$.
It must be observed that even though an aggregate action chosen requires knowledge of only the aggregate state,
the random distribution of energy (after a control is selected using \eqref{eqn:ql-sa}),
is achieved by knowing the exact buffer levels.

\begin{remark}
 An advantage of using state-aggregation with Q-learning is that it has convergence guarantees 
(Chapter 6, Section 6.2 \cite{bertsekas1996neuro}). This overcomes the problem of basis selection for function
approximation in the case of large state-action spaces. We have tried different partitoning schemes manually 
and all the schemes resulted in close policy performance. Also, we observed that increasing the number of
partitions improves the policy performance (see Fig. 6 in Section \ref{sec:results}) .
\end{remark}

\subsection{Cross Entropy using State Aggregation and Policy Parameterization}
\label{sec:cross-entropy}
The cross-entropy method is an iterative approach (\cite{rubinstein1999cross}) that we apply to find near-optimal stationary randomized policies for the energy sharing problem. 
The algorithm searches for a policy in the space of all stationary randomized policies in a systematic manner.
We define a class of randomized stationary policies $\{ \pi^{\boldsymbol{\theta}},\boldsymbol{\theta} \in \R^{M} \}$,
parameterized by a vector $\boldsymbol{\theta}$. For each pair $(s,a) \in S^{'} \times A^{'}$,
$\pi^{\boldsymbol{\theta}}(s,a)$ denotes the probability of taking action $a$ when the state $s$ is encountered
under the policy corresponding to $\boldsymbol{\theta}$.
In order to follow the cross entropy approach and obtain the optimal ${\boldsymbol{\theta}}^{*} \in \R^{M}$,
we treat each component $\theta_i, ~i \in \{ 1,2,\ldots, M \}$ of $\boldsymbol{\theta}$ as a normal random variable with mean $\mu_i$ and
variance $\sigma_i$. We will refer to these two quantities (the parameters of the normal distribution) as \emph{meta-parameters}.
We will tune the meta-parameters using the cross-entropy update rule \eqref{eqn:cross-entropy-upd} to find the
best values of $\mu_i$ and $\sigma_i$ which will correspond to a mean of $\theta_i^*$ and a variance of zero. The cross entropy method works as
follows:
Multiple samples of $\boldsymbol{\theta}$ namely ${\boldsymbol{\theta}}^1,{\boldsymbol{\theta}}^2, \ldots,
{\boldsymbol{\theta}}^N$  are generated according to the normal distribution with the
current estimate of the meta-parameters. Each sampled ${\boldsymbol{\theta}}$ will then
correspond to a stationary randomized policy.  We compute the average cost
$\lambda(\boldsymbol{\theta})$ of an SRP determined by a sample ${\boldsymbol{\theta}}$ by running a simulation
trajectory with the policy parameter fixed with the sample ${\boldsymbol{\theta}}$. We
perform this average cost computation for all the sampled ${\boldsymbol{\theta}}^i, i \in \{1,2,\ldots,N \}$, i.e., 
we compute $\lambda({\boldsymbol{\theta}}^1), \lambda({\boldsymbol{\theta}}^{2}),
\ldots, \lambda({\boldsymbol{\theta}}^{N})$. We then update the current estimates of the
meta-parameters based on only those sampled ${\boldsymbol{\theta}}$'s (policies) whose
average cost is lower than a threshold level (see \eqref{eqn:cross-entropy-upd}).
\begin{remark}
The Cross Entropy method is an adaptive importance sampling \cite{bucklew2004introduction} technique. 
The specific distribution from which the parameter $\theta$ is sampled is known as the importance sampling distribution. 
The Gaussian distribution used as the importance sampling distribution yields analytical updation formulas \eqref{eqn:cross-entropy-upd} for the mean and variance parameters (see \cite{kroese2006cross}). 
For this reason, it is convenient to use the Gaussian vectors for the policy parameters. 
\end{remark}

\subsubsection{Policy Parameterization}
\label{sec:policy-parameterization}
Let $\lambda(\boldsymbol{\theta})$ be the average cost of the system when parameterized by $\boldsymbol{\theta} = (\theta_1,\ldots,\theta_M)^\top$.
An optimal policy $\boldsymbol{\theta}^{*}$ minimizes the average cost over all parameterizations. 
That is,
\begin{displaymath}
 \boldsymbol{\theta}^{*} = \amin {\boldsymbol{\theta} \in \R^{M}} \: \lambda(\boldsymbol{\theta}).
\end{displaymath}
An example of parameterized randomized policies, which we use for the experiments (involving state aggregation) 
in this paper are the parameterized Boltzmann policies having the following form:
\begin{equation}
 \pi^{\theta}(s,a) = \frac{e^{ \boldsymbol{\theta^\top}\boldsymbol{\phi}_{sa} }} {\sum\limits_{b \in A(s)} e^{ \boldsymbol{\theta^\top}\boldsymbol{\phi}_{sb} }}
\qquad \forall s \in S^{'},\; \forall a \in A^{'}(s),
\end{equation}
where $\phi_{sa}$ is an $M$-dimensional feature vector for the aggregated state-action tuple $(s,a)$ and $\phi_{sa} \in \R^{M}$.
The parameterized Boltzmann policies are often used in approximation techniques (\cite{bhatnagar2009natural,bhatnagar2013feature,abounadi2001learning,sutton1998introduction,NIPS1999_1713}) which deal with randomized policies.
\begin{remark}
The probability distribution over actions is parameterized by $\boldsymbol{\theta}$ in the cross entropy method.
Since actions in every state need to be explored, the distribution needs to assign a non-zero
probability for every action feasible in a state. Hence the probability distribution
must be chosen based on these requirements. The Boltzmann distribution for action selection
 fits these requirements and is a frequently used distribution 
in the literature (see \cite{sutton1998introduction,NIPS1999_1713}) on policy 
learning and approximation algorithms.
\end{remark}
As noted in the beginning of this subsection, the parameters $\theta_1,\ldots,\theta_M$ are samples from the 
distributions $N(\mu_i,\sigma_i), 1 \leq i \leq M$, i.e., $\theta_i \sim N(\mu_i,\sigma_i), \; \forall i$.
\subsubsection{Method}
\label{sec:ce-method}
Initially $M$ parameter tuples $\{(\mu^{1}_i,\sigma^{1}_i), \; 1 \leq i \leq M\}$ for the normal distribution are picked. 
The policy is approximated using the Boltzmann distribution. The method comprises of two phases.
In the first phase trajectories corresponding to sample ${\boldsymbol{\theta}}$s are simulated
and the average cost of each policy is computed. The second phase inolves updation of the meta parameters.
The algorithm proceeds as follows: \\ Let iteration index $t$ be set to $1$. \\
\emph{First Phase:}
\begin{enumerate}
\item Sample parameters $\boldsymbol{\theta}^1,\ldots,\boldsymbol{\theta}^N$ are drawn independently from the normal distributions 
$\{N(\mu^{t}_i,\sigma^{t}_i), \text{  } 1 \leq i \leq M\}$. 
For $1 \leq j \leq N$, $\boldsymbol{\theta}^j \in \R^{M \times 1}$ and $\theta^j_i$ is sampled from $N(\mu^{t}_i,\sigma^{t}_i)$.

\item A trajectory is simulated using probability distribution $\pi^{\boldsymbol{\theta}^{j}}(s,a), \; 1 \leq j \leq N$.
Hence at every aggregate state $s$ an aggregate action $a$ is picked according to $\pi^{\boldsymbol{\theta}^{j}}(s,.)$.
Once an aggregate action is chosen for a state, the energy distribution is carried out as described in Section \ref{subsubsec:en-distr}.

\item The average cost per step of trajectory $j$ is $\lambda(\boldsymbol{\theta}^{j})$
and is computed for the trajectory simulated using $\boldsymbol{\theta}^j$. 
By abuse of notation we denote $\lambda(\boldsymbol{\theta}^{j})$ as $\lambda_j$.
\setcounter{mycounter}{\theenumi}
\end{enumerate}
\emph{Second Phase:}
\begin{enumerate}
\setcounter{enumi}{\themycounter}
 \item A quantile value $\rho \in(0,1)$ is selected. 
 \item The average cost values are sorted in descending order. Let $\lambda_1,\ldots,\lambda_N$
be the sorted order. Hence $\lambda_1 \geq \ldots \geq \lambda_N$.
 \item The ${\lceil(1-\rho)\rceil N}^{th}$ average cost is picked as the threshold level. So, let $\hat{\lambda}_c = \lambda_{\lceil(1-\rho)\rceil N} $. 
 \item The meta-parameters $\{(\mu^{t}_i,\sigma^{t}_i), \text{ } 1 \leq i \leq M\}$ are updated (refer \cite{menache2005basis}) in this phase. In iteration $t$, the parameters are updated in the second phase in the following manner:
\begin{equation}
\label{eqn:cross-entropy-upd}
\begin{split}
 \mu^{(t+1)}_{\!i} &= \frac{ {\sum\limits_{j = 1}^{N}} I_{\{\lambda_j \leq \hat{\lambda}_c\} } \theta^j_{i}}
		{ { \sum\limits_{j = 1}^{N}} I_{\{\lambda_j \leq \hat{\lambda}_c\} }},
\\
 \sigma^{2^{(t+1)}}_{\!i} &= \frac{ {\sum\limits_{j = 1}^{N}} I_{\{\lambda_j \leq \hat{\lambda}_c\} } \left(\theta^j_{i} - \mu^{(t+1)}_{\!i}\right)^2 }
			      { {\sum\limits_{j = 1}^{N}} I_{\{\lambda_j \leq \hat{\lambda}_c\} }} .
\end{split}
\end{equation}
 \item Set $t = t +1$ .
\end{enumerate}
Steps $1$-$6$ are repeated until the variances of the distributions converge to zero. 
Let $\boldsymbol{\mu} = (\mu_1,\ldots,\mu_M)^\top$ be the 
vector of means of the converged distributions. The near-optimal SRP found by the algorithm is $\hat{\pi}$ where
 \[ \hat{\pi}(s,a) = \frac{ e^{  \boldsymbol{\mu}\boldsymbol{\phi}(s,a) } } 
		    {\sum\limits_{b \in A(s)} e^{ \boldsymbol{\mu}\boldsymbol{\phi}(s,b) }},	\qquad \forall s \in S^{'}, a \in A^{'}(s). \]
\begin{remark}
The computational complexity of the cross entropy method is dependent on the number of updations required to 
arrive at the optimal parameter vector and the dimension of the vector. For instance in the case of four nodes, 
with data and energy buffer sizes being $30$, the cross entropy method requires $10^3$ sample
trajectories for a hyperparameter $(\mu,\sigma)$ vector of dimension $50$. The parameter $\boldsymbol{\theta}$ is updated
over $10^3$ iterations to arrive at the optimal parameter vector.
\end{remark}

\begin{remark}
The heuristic cross-entropy algorithm solves hard optimization problems. It is an iterative scheme and requires 
multiple samples to arrive at the solution. In general one assumes that the parameter $\theta$ is unknown (non-random variable) 
and uses actor-critic architecture to obtain locally optimal policy. However, obtaining gradient estimates in 
actor-critic architecture is hard as it leads to large variance \cite{konda2003onactor}. On the other hand, in 
our work, we let the parameter $\theta$ be a random variable and assume probability distibution over $\theta$ 
with hyperparameter $(\mu,\sigma)$ and use cross-entropy method to tune the hyperparameters. Cross entropy 
method is simple to implement, parallelizable and does not require gradient estimates. To the best of our 
knowledge, we are the first to combine the cross-entropy with state aggregation and apply it to a real world
problem. In \cite{mannor2003cross}, the authors sampled from the entire transition probability matrix to calculate
the score function and tested on problems with only small state-action space.
\end{remark}

\section{Simulation Results}
\label{sec:results}
In this section we show simulation results for the energy sharing algorithms we described in Sections \ref{sec:2nodes-algo} and \ref{sec:approximation}. 
For the sake of comparison we implement the greedy heuristic method in the case when the function $g$
has a non-linear form.
Also, we implement Q-learning to learn optimal policies for the case where we consider the sum of the data at all nodes and the available energy
as the state. These methods are as follows:
\begin{enumerate}	
 \item Greedy: This method takes as input the level of data $q^i_{k}$ at all nodes and supplies the energy based on the requirement. 
Since $g(x)$ is the number of data bits that can be sent given $x$ bits of allocated energy, 
$g^{-1}(y)$ gives the amount of energy required to send $y$ bits of data.
Suppose the energy available in the source is $e_k$ at stage $k$. The greedy algorithm then provides $t_k$ units of energy, 
where \(t_k = \min \left(e_k,{\sum\limits_{i = 1}^n}  g^{-1}(q^i_{k})\right)\). The energy bits are then shared between the
nodes based on the proportion of the requirement of the nodes.

\item Combined Nodes Q-learning: The state considered here is the sum of the data at all nodes and the available energy. 
Let the state space be $S_c$ and action space be $A_c$. So state $s_k = \left(\sum\limits^n_{i = 1} q^{i}_{k},E_k\right)$.
The control specified is $t_k$ which is the total energy that needs to be distributed between the nodes. 
In contrast to the action space in Section \ref{sec:es-prob-as-an-mdp}, here the exact split is not decided upon. Instead, this method
finds the total optimal energy to be supplied.
The algorithm in Section \ref{sec:ql} is then used to learn the optimal policies for the state-action space described here.

\end{enumerate}
In the above described methods, after an action $t_k$ is selected, the proportion of data in the nodes is computed. 
Thus $p_i = \frac{q^{i}_{k}}{\sum\limits^n_{j = 1} q^{j}_{k}}, \; 1 \leq i \leq n$ is computed at time $k$, where $0 \leq p_i \leq 1$ and
$\sum\limits^n_{i = 1} p_i = 1$.
Each of the $t_k$ bits of energy is then shared based on these probabilities. Let $u_i$ be the number of bits provided
to node $i$. Then in the case of the greedy method, $\sum\limits^n_{i = 1} u_i = t_k$, while in the combined nodes Q-learning
method, $\sum\limits^n_{i = 1} u_i \leq t_k$.
\subsection{Experimental Setup}
\label{sec:exp-setup}
\begin{itemize}
\item The algorithms described in Section \ref{sec:2nodes-algo} are simulated with two nodes and an energy source.
We consider the following settings:
\begin{enumerate}
 \item For the case of jointly Markov data arrival and Markovian energy arrival processes, we consider energy buffer
 size of $20$ and data buffer size of $10$. The data arrivals evolve as: $X_k = AX_{k-1} + \boldsymbol{\omega}$, 
where $A$ is a fixed $2 \times 2$ matrix of coefficients and $\boldsymbol{\omega} = (\omega_1,\omega_2)^{\top}$ 
is a $2 \times 1$ random noise (or disturbance) vector. 
Here $A = \bigl(\begin{smallmatrix}
0.2&0.3&\\0.3&0.2
\end{smallmatrix}\bigr) $
The energy arrival evolves as $Y_k = bY_{k-1} + \chi$, where $\chi$ is also random noise (or disturbance) variable and  
$b = 0.5$ is a fixed coefficient. The components in vector $\boldsymbol{\omega}$ and $\chi$ are Poisson distributed.
In the simulations, we vary the mean of the random noise variable $\omega_1$, while means of $\omega_2$, $\chi$ are kept constant.
 \item For the case of i.i.d data and energy arrivals the data and energy buffer sizes are fixed at $14$.
$X^1, \, X^2 \text{ and } \, Y$ are distributed according to the Poisson distribution.
In the simulations, the mean data arrival at node two is fixed while that at node one is varied.
\end{enumerate}
\item The algorithms described in Section \ref{sec:approximation} are simulated with four nodes and an energy source.
We consider the following settings:
\begin{enumerate}
 \item For the case of jointly Markov data arrival and Markovian energy arrival processes, we consider energy buffer
  size of $25$ and data buffer size of $10$. The data arrivals evolve as: 
$X_k = AX_{k-1} + \boldsymbol{\omega}$, 
where $A$ is a fixed $4 \times 4$ matrix of coefficients and $\boldsymbol{\omega} = (\omega_1,\omega_2,\omega_3,\omega_4)^{\top}$ 
is a $4 \times 1$ random noise (or disturbance) vector. 
Here \\$A = \biggl(\begin{smallmatrix}
0.1&0.1&0.1&0.2\\ 0.1&0.1&0.2&0.1\\0.1&0.2&0.1&0.1\\0.2&0.1&0.1&0.1
\end{smallmatrix}\biggr) $
The energy arrival evolves as $Y_k = bY_{k-1} + \chi$, where $\chi$ is also random noise (or disturbance) variable and  
$b = 0.5$ is a fixed coefficient. The components in vector $\boldsymbol{\omega}$ and $\chi$ are Poisson distributed.
In the simulations, we vary the mean of the random noise variable $\omega_1$, while means of $\omega_2$-$\omega_4$, $\chi$ are kept constant.
The energy buffer had $4$ partitions while the data buffer had $2$ partitions.
 \item For the case of i.i.d data and energy arrivals, the buffer sizes are taken to be $30$ each. 
The data and energy buffers are clustered into $6$ partitions. 
$X^1,X^2,X^3,X^4,Y$ are Poisson distributed. 
In these experiments, the mean data arrivals at nodes $2,3 \text{ and }4$ are fixed while the same at node $1$ is varied.
\end{enumerate}
\end{itemize}
For all Q-learning algorithms ($\epsilon$-greedy, UCB based and Combined Nodes), 
stepsize $\alpha=0.1$ is used in the updation scheme. For the $\epsilon$-greedy method, $\epsilon = 0.1$
is used for exploration. In the UCB exploration mechanism, the value of $\beta$ is set to $1$. In our experimental simulations,
we consider the function $g(x) = \ln(1+x)$ for the i.i.d case and $g(x) = 2\ln(1+x)$ for the non-i.i.d case.

\subsection{Results}
Figs. \ref{plot:2n-non-iid}, \ref{plot:2n-demax14}, \ref{plot:2n-hexp-demax14}, \ref{plot:2n-demax30}
and \ref{plot:2n-r2-0.3} show the performance of the algorithms explained in 
Section \ref{sec:2nodes-algo}. The simulations are carried out with
two nodes and a single source. Similarly, Figs. \ref{plot:4n-non-iid}, \ref{plot:4n-basic} and \ref{plot:4n-hyperexp} show the 
performance comparisons of our algorithms explained in Section \ref{sec:approximation} with other algorithms.
The simulations in this case are carried out with four nodes and a single source.
In Figs. \ref{plot:2n-non-iid}, \ref{plot:4n-non-iid} jointly Markov data arrival and Markovian energy arrival 
processes are considered and the noise in data and energy arrival at Node 1, i.e. $\E {\omega_1}$ is varied
while that at the other nodes is kept constant. The i.i.d case of data and energy arrivals is considered
in Figs. \ref{plot:2n-demax14},\ref{plot:2n-hexp-demax14}, \ref{plot:2n-demax30}, \ref{plot:4n-basic} and \ref{plot:4n-hyperexp}.
In these plots, the mean data arrival at Node 1 ($\E {X^1}$) is 
varied while keeping that at the other node(s) constant. 
Figs. \ref{plot:2n-non-iid}-\ref{plot:2n-ql-qlsa-gr-cnql} show
the normalized long-run average cost of the policies determined by the algorithms along the y-axis. 
The mean energy arrival is also fixed.

The Q-learning algorithm is designed to learn optimal policies,
hence it outperforms other algorithms, as shown in Figs. \ref{plot:2n-non-iid},\ref{plot:2n-demax14}, 
\ref{plot:2n-hexp-demax14}, \ref{plot:2n-demax30} and \ref{plot:2n-r2-0.3}.
The policy learnt by our algorithm does better compared to the greedy policy and the policy obtained from the combined nodes Q-learning method.
Note that Q-learning on combined nodes learns the total energy to be distributed and not the exact split.
Hence its performance is poor compared to Q-learning on our problem MDP. Thus, sharing energy by considering the total amount of data 
in all the nodes is not optimal.
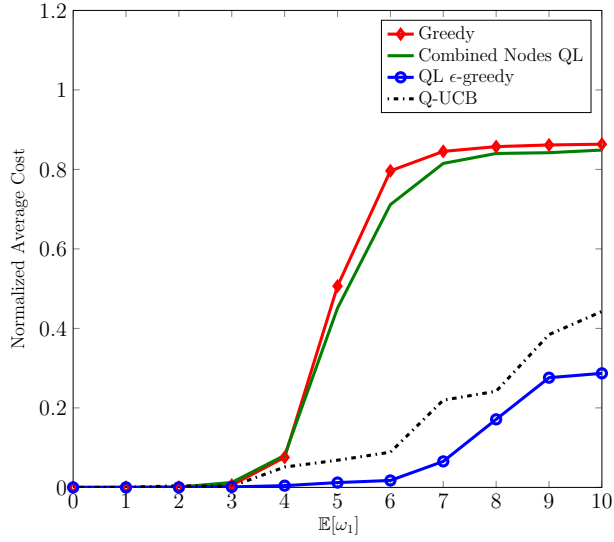
\begin{figure}
 \begin{center}
   \scalebox{0.58}{
%
%
%
\definecolor{mycolor1}{rgb}{0.04314,0.51765,0.78039}%
\definecolor{mycolor2}{rgb}{0.00000,0.49804,0.00000}%
\begin{tikzpicture}

\begin{axis}[%
width=4.7715625in,
height=4.3076875in,
scale only axis,
xmin=0,
xmax=10,
xtick={0,1,2,3,4,5,6,7,8,9,10},
xlabel={$\E{\omega_1}$},
ymin=0.000000115,
ymax=1.2,
ylabel={Normalized Average Cost},
legend style={draw=black,fill=white,legend cell align=left}
]

\addplot [color=red,solid,line width=2.0pt,mark=diamond,mark options={solid},mark size=3pt] table[row sep=crcr]
{
0        2.50e-007  	\\
1	 2.09e-006  	\\
2	 0.000311055	\\
3	 0.00573778 	\\
4	 0.0756285  	\\
5	 0.506367   	\\
6	 0.796353   	\\
7	 0.845342   	\\
8	 0.857313   	\\
9	 0.86151    	\\
10	 0.863422  	\\
};
\addlegendentry{Greedy};

\addplot [color=mycolor2,solid,line width=2.0pt] table[row sep=crcr]
{
0  	 1.20e-007     	\\
1	 0.00014819   	\\
2	 0.00094484   	\\
3	 0.0123985    	\\
4	 0.0805757    	\\
5	 0.450983     	\\
6	 0.711415     	\\
7	 0.815113     	\\
8	 0.840066     	\\
9	 0.842051     	\\
10	 0.848643      	\\
};
\addlegendentry{Combined Nodes QL};

\addplot [color=blue,solid,line width=2.0pt,mark=o,mark options={solid},mark size=3pt] table[row sep=crcr]
{
0 	2.20e-007      	\\
1	3.02e-006	\\
2	0.000146315	\\
3	0.00122754	\\
4	0.00440182	\\
5	0.012111	\\
6	0.0174645	\\
7	0.0654471	\\
8	0.171207	\\
9	0.276		\\
10	0.286947	\\
};
\addlegendentry{QL $\epsilon$-greedy};

\addplot [color=black,dash pattern=on 1pt off 3pt on 3pt off 3pt,line width=2.0pt]
  table[row sep=crcr]{
0 	1.15e-007      	\\
1	0.000212	\\
2	0.00402185	\\
3	0.0029398	\\
4	0.0511799	\\
5	0.0684		\\
6	0.0885		\\
7	0.21948		\\
8	0.241853	\\
9	0.384456	\\
10	0.442614	\\
};
\addlegendentry{Q-UCB};

\end{axis}
\end{tikzpicture}%
}
    \caption{$E_{\!{MAX}} = 20,\,D_{\!{MAX}} = 10$, $\omega_1,\,\omega_2,\,\chi$ are Poisson distributed 
		 with $\E {\omega_2} = 1.0$,$\E {\chi} = 20$, } 
     \label{plot:2n-non-iid}
 \end{center}
\end{figure}

\begin{figure}    
    \begin{subfigure}[t]{0.5\columnwidth}	
      \begin{center}
        \scalebox{0.5}{
%
%
%
\definecolor{mycolor1}{rgb}{0.04314,0.51765,0.78039}%
\definecolor{mycolor2}{rgb}{0.00000,0.49804,0.00000}%
\begin{tikzpicture}

\begin{axis}[%
width=4.7715625in,
height=4.3076875in,
scale only axis,
xmin=0,
xmax=2.5,
xtick={0,0.25,0.5,0.75,1.0,1.25,1.5,1.75,2.0,2.25,2.5},
xlabel={$\E {X^1}$},
ymin=0.4,
ymax=2.4,
ylabel={Normalized Average Cost},
legend style={draw=black,fill=white,legend cell align=left}
]

\addplot [color=red,solid,line width=2.0pt,mark=diamond,mark options={solid},mark size=3pt] table[row sep=crcr]
{
0.0     0.460701\\
0.25	0.50831\\
0.5	0.58602\\
0.75	0.7306\\
1	0.99164\\
1.25	1.32357\\
1.5	1.54487\\
1.75	1.6466\\
2	1.6944\\
2.25	1.71966\\
2.5	1.73438\\
};
\addlegendentry{Greedy};

\addplot [color=mycolor2,solid,line width=2.0pt] table[row sep=crcr]
{
0.0     0.454837
0.25	0.487531\\
0.5	0.56333\\
0.75	0.643118\\
1	0.88301\\
1.25	0.885\\
1.5	1.17746\\
1.75	1.41471\\
2	1.53603\\
2.25	1.59425\\
2.5	1.62519\\
};
\addlegendentry{Combined Nodes QL};

\addplot [color=blue,solid,line width=2.0pt,mark=o,mark options={solid},mark size=3pt] table[row sep=crcr]
{
0.0	0.452877\\
0.25	0.4684\\
0.5	0.487809\\
0.75	0.5091\\
1	0.52\\
1.25	0.601472\\
1.5	0.817782\\
1.75	1.08151\\
2	1.28094\\
2.25	1.34175\\
2.5	1.35916\\
};
\addlegendentry{QL $\epsilon$-greedy};


\addplot [color=black,dash pattern=on 1pt off 3pt on 3pt off 3pt,line width=2.0pt]
  table[row sep=crcr]{
0.0     0.451632\\
0.25	0.46019\\
0.5	0.47223\\
0.75	0.483316\\
1	0.505403\\
1.25	0.574661\\
1.5	0.717441\\
1.75	0.960172\\
2	1.21124\\
2.25	1.297\\
2.5	1.31592\\
};
\addlegendentry{Q-UCB};


\end{axis}
\end{tikzpicture}%
}	
	 \caption{$X^1,\,X^2,\,Y$ are Poisson distributed with $\E Y = 13$, $\E {X^2} = 1.0$ }	
	\label{plot:2n-demax14}
      \end{center}
    \end{subfigure}~
        ~ 
    \begin{subfigure}[t]{0.5\columnwidth}
	\begin{center}
        \scalebox{0.5}{
%
%
%
\definecolor{mycolor1}{rgb}{0.04314,0.51765,0.78039}%
\definecolor{mycolor2}{rgb}{0.00000,0.49804,0.00000}%
\begin{tikzpicture}

\begin{axis}[%
width=4.7715625in,
height=4.3076875in,
scale only axis,
xmin=0,
xmax=3,
xtick={0,0.25,0.5,0.75,1.0,1.25,1.5,1.75,2.0,2.25,2.5,2.75,3.0},
xlabel={$\E {X^1}$},
ymin=0.4,
ymax=2.4,
ylabel={Normalized Average Cost},
legend style={draw=black,fill=white,legend cell align=left}
]
\addplot [color=red,solid,line width=2.0pt,mark=diamond,mark options={solid},mark size=3pt] table[row sep=crcr]
{
0.0   0.839145\\
0.25	0.983704\\
0.5	1.19488\\
0.75	1.49\\
1	1.69173\\
1.25	1.77456\\
1.5	1.8107\\
1.75	1.829\\
2	1.84067\\
2.25	1.85776\\
2.5	1.86265\\
2.75	1.876\\
3	1.898\\
};
\addlegendentry{Greedy};

\addplot [color=mycolor2,solid,line width=2.0pt] table[row sep=crcr]
{
0.0     0.5799\\
0.25	0.6563\\
0.5	0.844\\
0.75	0.8672\\
1	0.927\\
1.25	1.06304\\
1.5	1.29\\
1.75	1.46252\\
2	1.48207\\
2.25	1.5617\\
2.5	1.66521\\
2.75	1.67115\\
3	1.80174\\
};
\addlegendentry{Combined Nodes QL};

\addplot [color=blue,solid,line width=2.0pt,mark size=3pt,mark=o,mark options={solid,fill=black}]
  table[row sep=crcr]{0.0    0.52363\\
0.25	0.61\\
0.5	0.68314\\
0.75	0.76\\
1	0.85\\
1.25	0.90615\\
1.5	1.07435\\
1.75	1.25413\\
2	1.373\\
2.25	1.497\\
2.5	1.55012\\
2.75	1.61741\\
3	1.713\\
};
\addlegendentry{QL $\epsilon$-greedy};

\addplot [color=black,dash pattern=on 1pt off 3pt on 3pt off 3pt,line width=2.0pt]
  table[row sep=crcr]{
0.0     0.504847\\
0.25	0.582083\\
0.5	0.642382\\
0.75	0.696701\\
1	0.77869\\
1.25	0.85472\\
1.5	1.04473\\
1.75	1.1082\\
2	1.2008\\
2.25	1.21196\\
2.5	1.50078\\
2.75	1.52552\\
3.0	1.70201\\
};
\addlegendentry{Q-UCB};


\end{axis}
\end{tikzpicture}%
}        
         \caption{$X^1$: Poisson distributed, $X^2$: hyperexponential distributed and $Y$: Exponential distributed and  $\E {X^2} = 0.625$, $\E Y = 10 $ }       
        \label{plot:2n-hexp-demax14}
      \end{center}
    \end{subfigure}
    \caption{Performance comparison of policies when $E_{\!{MAX}} = D_{\!{MAX}} = 14$}
\end{figure}
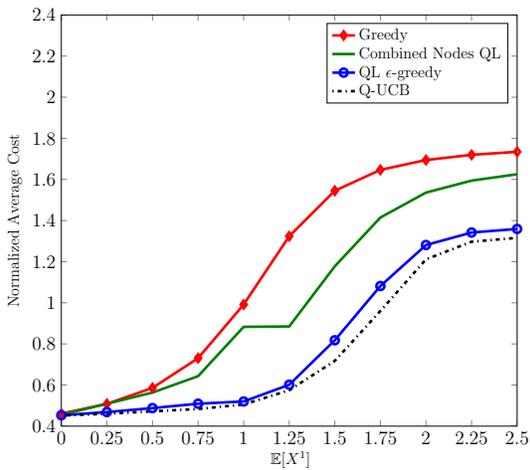
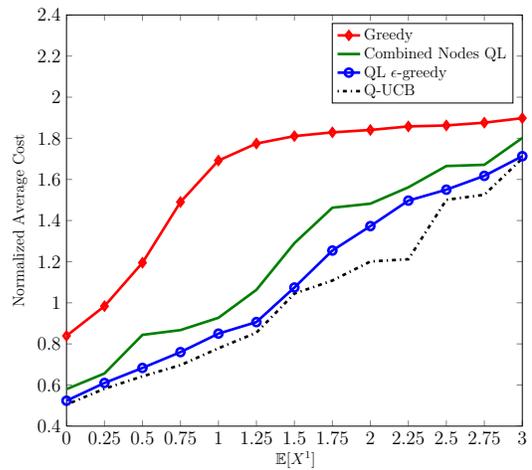

\begin{figure}[h]
\begin{center}
\scalebox{0.55}{
%
%
%
\definecolor{mycolor1}{rgb}{0.04314,0.51765,0.78039}%
\definecolor{mycolor2}{rgb}{0.00000,0.49804,0.00000}%
\begin{tikzpicture}

\begin{axis}[%
width=5.7715625in,
height=4.3076875in,
scale only axis,
xmin=0,
xmax=3.0,
xtick={0,0.25,0.5,0.75,1.0,1.25,1.5,1.75,2.0,2.25,2.5,2.75,3.0},
xlabel={$\E {X^1}$},
ymin=0,
ymax=2.2,
ylabel={Normalized Average Cost},
legend style={draw=black,fill=white,legend cell align=left}
]

\addplot [color=red,solid,line width=2.0pt,mark=diamond,mark options={solid},mark size=3pt] table[row sep=crcr]
{
0.0    0.0313\\
0.25	0.05779\\
0.5	0.0721\\
0.75	0.1068\\
1	0.193\\
1.25	0.3506\\
1.5	0.53748\\
1.75	0.7555\\
2	1.1338\\
2.25	1.146\\
2.5	1.5635\\
2.75	1.6\\
3	1.6467\\
};
\addlegendentry{Greedy};

\addplot [color=mycolor2,solid,line width=2.0pt] table[row sep=crcr]
{
0.0     0.005902\\
0.25	0.022\\
0.5	0.03015\\
0.75	0.041586\\
1	0.0689\\
1.25	0.21716\\
1.5	0.476541\\
1.75	0.748655\\
2	0.92657\\
2.25	1.11321\\
2.5	1.32579\\
2.75	1.49158\\
3	1.62603\\
};
\addlegendentry{Combined Nodes QL};

\addplot [color=blue,solid,line width=2.0pt,mark=o,mark options={solid},mark size=3pt] table[row sep=crcr]
{
0.0    0.004\\
0.25	0.010656\\
0.5	0.0129\\
0.75	0.03486\\
1	0.04336\\
1.25	0.13525\\
1.5	0.31058\\
1.75	0.70355\\
2	0.8558\\
2.25	0.96741\\
2.5	0.9819\\
2.75	0.9857\\
3	1.12592\\
};
\addlegendentry{QL $\epsilon$-greedy};

\addplot [color=black,dash pattern=on 1pt off 3pt on 3pt off 3pt,line width=2.0pt]
  table[row sep=crcr]{
0.0     0.0015226\\
0.25	0.007632\\
0.5	0.0105\\
0.75	0.0175\\
1	0.03105\\
1.25	0.0595355\\
1.5	0.08033\\
1.75	0.11136\\
2	0.2839\\
2.25	0.706731\\
2.5	0.87058\\
2.75	0.964\\
3	1.01043\\
};
\addlegendentry{Q-UCB};
%

\end{axis}
\end{tikzpicture}%
}
\caption{Performance comparison of policies with $E_{\!{MAX}} = D_{\!{MAX}} = 30$ when $X^1,\,X^2,\,Y$ are Poisson distributed with $\E Y = 25, \E {X^2} = 1.0$ }
\label{plot:2n-demax30}
\end{center}
\end{figure}
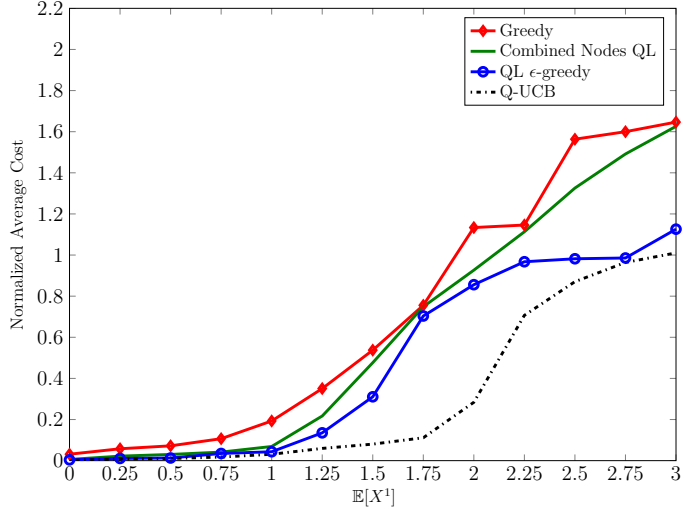
Figs. \ref{plot:4n-basic} and \ref{plot:4n-hyperexp} show the long-run normalized average costs of the policies obtained from the Greedy method and the 
algorithms described in Section \ref{sec:approximation}. Since our algorithms are model-free, irrespective of the distributions of energy and data
arrival (see Figs. \ref{plot:4n-basic} and \ref{plot:4n-hyperexp}), our algorithms learn the optimal or near-optimal policies.
These plots show that our approximation algorithms outperform the greedy and combined nodes Q-learning methods.
It can be observed that the gap between the average costs obtained from the combined nodes Q-learning method and the approximate learning
algorithm (see Section \ref{sec:approx-alg}) increases with an increase in the number of nodes. This is clear from  
Figs. \ref{plot:2n-demax14} and \ref{plot:4n-basic}. This occurs because the combined nodes Q-learning method wastes
energy and the amount of wastage increases with an increase in the number of nodes.
\begin{figure}[h]
\begin{center} 
        \scalebox{0.58}{
%
%
%
\definecolor{mycolor1}{rgb}{0.00000,0.49804,0.00000}%
\definecolor{mycolor2}{rgb}{0.8,0.3,0.1}%
\begin{tikzpicture}

\begin{axis}[%
width=4.82222222222222in,
height=3.80333333333333in,
scale only axis,
xmin=0,
xmax=10,
xtick={0,1,2,3,4,5,6,7,8,9,10},
xlabel={$\E{\omega_1}$},
ymin=0,
ymax=4.0,
ylabel={Normalized Average Cost},
legend style={draw=black,fill=white,legend cell align=left}
]

\addplot [color=mycolor1,dash pattern=on 1pt off 3pt on 3pt off 3pt,line width=2.0pt] 
  table[row sep=crcr]{
0 	9.50E-004    	\\
1	0.0053468	\\
2	0.190668	\\
3	1.25		\\
4	1.98824		\\
5	2.23342		\\
6	2.33247		\\
7	2.37502		\\
8	2.39375		\\
9	2.40288		\\
10	2.40755		\\
};
\addlegendentry{Greedy};

\addplot [color=mycolor2,solid,line width=2.0pt] 
  table[row sep=crcr]{
0 	0.000955635    	\\
1	0.0210784	\\
2	0.407937	\\
3	1.67914		\\
4	2.26454		\\
5	2.45073		\\
6	2.5233		\\
7	2.5505		\\
8	2.56045		\\
9	2.56543		\\
10	2.56133		\\
};
\addlegendentry{Combined Nodes QL};

\addplot [color=blue,dashed,line width=2.0pt,mark=o,mark options={solid},mark size=3.1pt]  table[row sep=crcr]
{
0 	0.0001  \\
1	0.004123\\
2	0.1834  \\
3	1.23202 \\
4	1.8345  \\
5	2.1     \\
6	2.234   \\
7	2.303   \\
8 	2.3598	\\
9	2.3724	\\
10	2.3775	\\
};
\addlegendentry{Cross Entropy};

\addplot [color=black,solid,line width=2.0pt,mark=+,mark options={solid},mark size=4pt]
  table[row sep=crcr]{
0 	0.0009    	\\
1	0.003133	\\
2	0.176634	\\
3	1.13472		\\
4	1.70093		\\
5	1.85774		\\
6	1.93827		\\
7	1.97		\\
8	1.99433		\\
9	2.00347		\\
10	2.00679		\\
};
\addlegendentry{QL-SA $\epsilon-$greedy};

\addplot [color=red,solid,line width=2.0pt,mark=diamond,mark size=3pt,mark options={solid}]
  table[row sep=crcr]{
0 	0.00093    	\\
1	0.0032		\\
2	0.18		\\
3	1.14		\\
4	1.70119		\\
5	1.9		\\
6	1.95		\\
7	1.977		\\
8	1.995		\\
9	2.02		\\
10	2.05		\\
};
\addlegendentry{QL-SA-UCB};

\end{axis}
\end{tikzpicture}%
}        
\end{center}
\caption{$E_{\!{MAX}} = 25,\,D_{\!{MAX}} = 10$, $\omega_1$-$\omega_5$ are Poisson distributed 
		 with $\E {\omega_2} = \E {\omega_3} = \E {\omega_4} = 1.0$, $\E {\chi} = 5$}
        \label{plot:4n-non-iid}
    \end{figure}
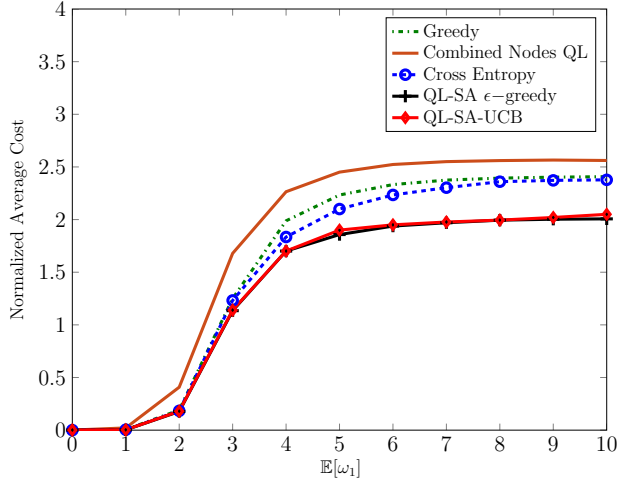
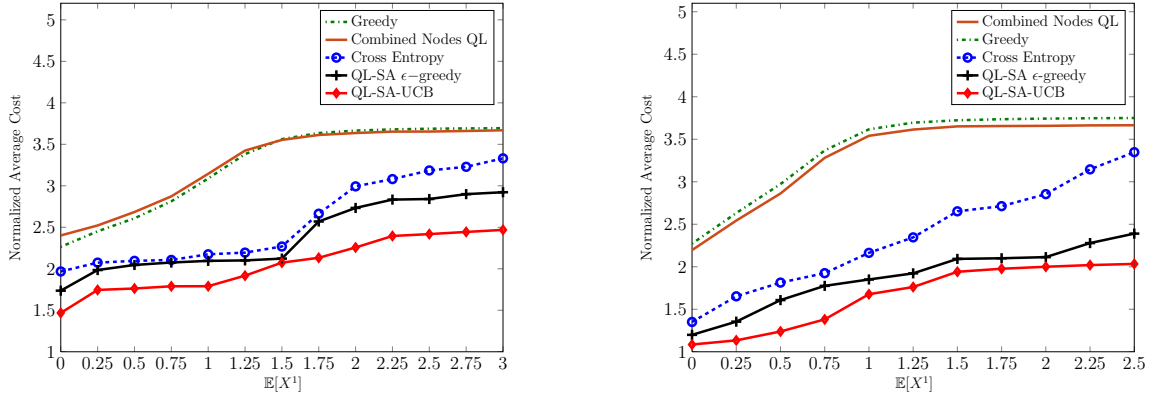
\begin{figure}
     \begin{subfigure}[t]{0.5\columnwidth}
      \begin{center}
       \scalebox{0.48}{
%
%
%
\definecolor{mycolor1}{rgb}{0.00000,0.49804,0.00000}%
\definecolor{mycolor2}{rgb}{0.8,0.3,0.1}%
\begin{tikzpicture}

\begin{axis}[%
width=4.82222222222222in,
height=3.80333333333333in,
scale only axis,
xmin=0,
xmax=3.0,
xtick={0,0.25,0.5,0.75,1.0,1.25,1.5,1.75,2.0,2.25,2.5,2.75,3.0},
xlabel={$\E {X^1}$},
ymin=1.0,
ymax=5.2,
ylabel={Normalized Average Cost},
legend style={draw=black,fill=white,legend cell align=left}
]

\addplot [color=mycolor1,dash pattern=on 1pt off 3pt on 3pt off 3pt,line width=2.0pt] 
  table[row sep=crcr]{
0.0     2.26464\\
0.25	2.45178\\
0.5	2.60852\\
0.75	2.81272\\
1	3.08813\\
1.25	3.38143\\
1.5	3.56286\\
1.75	3.63602\\
2	3.66554\\
2.25	3.67958\\
2.5	3.688\\
2.75	3.69142\\
3.0     3.69423\\
};
\addlegendentry{Greedy};

\addplot [color=mycolor2,solid,line width=2.0pt] 
  table[row sep=crcr]{
0.0     2.40182\\
0.25	2.52209\\
0.5	2.68431\\
0.75	2.87283\\
1	3.14467\\
1.25	3.4244\\
1.5	3.55307\\
1.75	3.61262\\
2	3.63606\\
2.25	3.65276\\
2.5	3.65374\\
2.75	3.66073\\
3.0     3.669\\
};
\addlegendentry{Combined Nodes QL};

\addplot [color=blue,dashed,line width=2.0pt,mark=o,mark options={solid},mark size=3.1pt] 
  table[row sep=crcr]{
0.0     1.9678\\
0.25	2.07415\\
0.5	2.09551\\
0.75	2.106541\\
1	2.17466\\
1.25	2.19491\\
1.5	2.2685\\
1.75	2.66373\\
2	2.99332\\
2.25	3.081\\
2.5	3.185\\
2.75	3.229\\
3.0     3.33073\\
};
\addlegendentry{Cross Entropy};

\addplot [color=black,solid,line width=2.0pt,mark=+,mark options={solid},mark size=4pt]
  table[row sep=crcr]{
0.0     1.73641\\
0.25	1.98495\\
0.5	2.0461\\
0.75	2.07561\\
1	2.09549\\
1.25	2.1009\\
1.5	2.12272\\
1.75	2.56983\\
2	2.73334\\
2.25	2.83382\\
2.5	2.84077\\
2.75	2.89939\\
3.0     2.92263\\
};
\addlegendentry{QL-SA $\epsilon-$greedy};

\addplot [color=red,solid,line width=2.0pt,mark=diamond,mark size=3pt,mark options={solid}]  table[row sep=crcr]
{
0.0     1.46796\\
0.25	1.74529\\
0.5	1.7623\\
0.75	1.78915\\
1	1.79\\
1.25	1.9173\\
1.5	2.07418\\
1.75	2.13246\\
2	2.2568\\
2.25	2.39512\\
2.5	2.41812\\
2.75	2.444\\
3.0     2.469\\
};
\addlegendentry{QL-SA-UCB};

\end{axis}
\end{tikzpicture}%
}
	\caption{$X^1,\,X^2,\,X^3,\,X^4,Y$ are Poisson distributed with 
		  $\E {X^2} = \E {X^3} = \E {X^4} = 1.0$, $\E Y = 25$ }
	\label{plot:4n-basic}
	\end{center}
    \end{subfigure}~
    \begin{subfigure}[t]{0.5\columnwidth}
	\begin{center}
        \scalebox{0.48}{
%
%
%
\definecolor{mycolor1}{rgb}{0.00000,0.49804,0.00000}%
\definecolor{mycolor2}{rgb}{0.8,0.3,0.1}%
\begin{tikzpicture}

\begin{axis}[%
width=4.82222222222222in,
height=3.80333333333333in,
scale only axis,
xmin=0,
xmax=2.5,
xtick={0,0.25,0.5,0.75,1.0,1.25,1.5,1.75,2.0,2.25,2.5},
xlabel={$\E {X^1}$},	
ymin=1.0,
ymax=5.1,
ylabel={Normalized Average Cost},
legend style={draw=black,fill=white,legend cell align=left}
]

\addplot [color=mycolor2,solid,line width=2.0pt] table[row sep=crcr]
{
0.0     2.196\\
0.25	2.54327\\
0.5	2.86187\\
0.75	3.28141\\
1	3.54063\\
1.25	3.61407\\
1.5	3.65144\\
1.75	3.65555\\
2	3.6577\\
2.25	3.66312\\
2.5	3.66483\\
};
\addlegendentry{Combined Nodes QL};

\addplot [color=mycolor1,dash pattern=on 1pt off 3pt on 3pt off 3pt,line width=2.0pt]  table[row sep=crcr]
{
0.0     2.27216\\
0.25	2.63217\\
0.5	2.97254\\
0.75	3.36831\\
1	3.61671\\
1.25	3.69542\\
1.5	3.72326\\
1.75	3.73586\\
2	3.74285\\
2.25	3.74712\\
2.5	3.75\\
};
\addlegendentry{Greedy};

\addplot [color=blue,dashed,line width=2.0pt,mark=o,mark options={solid},mark size=3.1pt] table[row sep=crcr]
{
0.0     1.35052\\
0.25	1.65214\\
0.5	1.81365\\
0.75	1.9247\\
1	2.16355\\
1.25	2.34560\\
1.5	2.65214\\
1.75	2.71305\\
2	2.85439\\
2.25	3.14576\\
2.5	3.3478\\
};
\addlegendentry{Cross Entropy};

\addplot [color=black,solid,line width=2.0pt,mark=+,mark options={solid},mark size=4pt] table[row sep=crcr]
{
0.0       1.19895\\
0.25	  1.35417\\
0.5	  1.60781\\
0.75	  1.77532\\
1.0	  1.85047\\
1.25	  1.9224\\
1.5	  2.0925\\
1.75	  2.10005\\
2.0	  2.11334\\
2.25	  2.27867\\
2.5	  2.39\\
};
\addlegendentry{QL-SA $\epsilon$-greedy};

\addplot [color=red,solid,line width=2.0pt,mark=diamond,mark size=3pt,mark options={solid}]  table[row sep=crcr]
{
0.0       1.08511\\
0.25	  1.13413\\
0.5	  1.23847\\
0.75	  1.38082\\
1.0	  1.67722\\
1.25	  1.76163\\
1.5	  1.94118\\
1.75	  1.97678\\
2.0	  2.0\\
2.25	  2.01981\\
2.5	  2.03354\\
};
\addlegendentry{QL-SA-UCB};

\end{axis}
\end{tikzpicture}%
}
        \caption{$X^1,\,X^3,\,X^4$ are Poisson distributed with $\E {X^3} = \E {X^4} = 0.7$, $X^2$ is 
		  hyperexponentially distributed and $Y$ has the exponential distribution with $\E Y = 20$ }
        \label{plot:4n-hyperexp}
	\end{center}
    \end{subfigure}
    \caption{Performance comparison of policies when $E_{\!{MAX}} = D_{\!{MAX}} = 30$}
\end{figure}

Fig. \ref{plot:4n-sa-levels} shows the variation in average cost with different number of partitions of data and energy buffers
used in state aggregation. As the number of partitions increase, the number of clusters also increase resulting in better policies.
\begin{figure}[h]
 \begin{center}
 \scalebox{0.65}{
%
%
%
\definecolor{mycolor1}{rgb}{0.00000,0.49804,0.00000}%
\begin{tikzpicture}

\begin{axis}[%
width=4.82222222222222in,
height=4.00333333333333in,
scale only axis,
xmin=0,
xmax=3.0,
xtick={0,0.25,0.5,0.75,1.0,1.25,1.5,1.75,2.0,2.25,2.5,2.75,3.0},
xlabel={$\E {X^1}$},
ymin=1.65,
ymax=3.4,
ylabel={Normalized Average Cost},
legend style={draw=black,fill=white,legend cell align=left},
legend pos= north west
]
\addplot [color=mycolor1,solid,line width=2.0pt,mark size=2.5pt,mark=o,mark options={solid}]   table[row sep=crcr]
{
0.0     1.95131\\
0.25	2.02777\\
0.5	2.08802\\
0.75	2.08354\\
1	2.13446\\
1.25	2.185\\
1.5	2.2545\\
1.75	2.641\\
2	2.92\\
2.25	2.993\\
2.5	3.0102\\
2.75	3.0135\\
3.0     3.01216\\
};
\addlegendentry{QL-SA $\epsilon$-greedy (3 partitions)};

\addplot [color=blue,dashed,line width=2.0pt,mark size=2.5pt,mark=square,mark options={solid}]  table[row sep=crcr]
{
0.0     1.94562\\
0.25	2.01175\\
0.5	2.06065\\
0.75	2.08152\\
1	2.1321\\
1.25	2.17141\\
1.5	2.20241\\
1.75	2.60924\\
2	2.88373\\
2.25	2.9887\\
2.5	3.00086\\
2.75	3.00251\\
3.0     3.00755\\
};
\addlegendentry{QL-SA $\epsilon$-greedy (4 partitions)};

\addplot [color=black,solid,line width=2.0pt]
  table[row sep=crcr]{
0.0     1.73641\\
0.25	1.98495\\
0.5	2.04609\\
0.75	2.07561\\
1	2.09549\\
1.25	2.1009\\
1.5	2.12272\\
1.75	2.56983\\
2	2.73334\\
2.25	2.83382\\
2.5	2.84077\\
2.75	2.8994\\
3.0     2.92263\\
};
\addlegendentry{QL-SA $\epsilon$-greedy (6 partitions)};

\end{axis}
\end{tikzpicture}%
}
 \end{center}
 \caption{Performance of QL $\epsilon$-greedy with different number of data and energy buffer partitions, when $X^1,\,X^2,\,X^3,\,X^4,Y$ are 
Poisson distributed with  $\E {X^2} = \E {X^3} = \E{X^4} = 1.0$, $\E{Y} = 25$ and $D_{\!{MAX}} = E_{\!{MAX}} = 30$}
 \label{plot:4n-sa-levels}
 \end{figure}
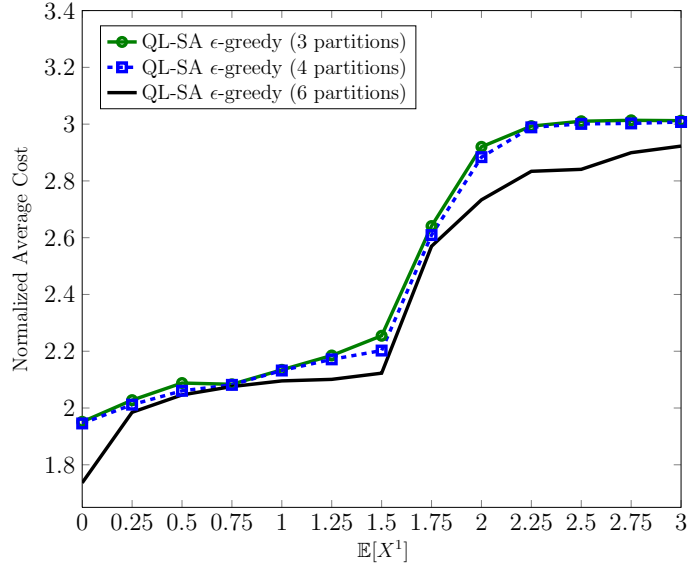

The single-stage cost function as defined in \eqref{eqn:cost-fn}, includes the effect of action in the conversion function $g(\cdot)$.
The effect of the action taken can be explicitly included in the single-stage cost function of the following form:
\begin{equation}
\label{eqn:cost-function-with-r2}
c(s_k,T(s_k)) = {\sum_{i = 1}^{n}} (r_1*(q^i_{k} - g(T^i(s_{\!k})))^{+} + r_2*T^i(s_{\!k}) ),
\end{equation}
where $r_1$, $r_2$ are the tradeoff parameters, $r_1+r_2 = 1$ and $r_1,r_2 \geq 0$. The above equation is a convex combination
of the sum of data queue lengths and the collective energy supplied to the nodes. It can be observed that the single-stage 
cost function \eqref{eqn:cost-fn} used in our MDP model can be derived from \eqref{eqn:cost-function-with-r2} by taking
$r_1 = 1$ and $r_2 = 0$. When $r_1 > 0$ and $r_2 > 0$, the cost structure \eqref{eqn:cost-function-with-r2}
gives importance to the data queue length as well as the amount of energy supplied. The performance comparison of 
our algorithms (described in Section \ref{sec:2nodes-algo}) with the greedy and combined nodes Q-learning methods 
using this single-stage cost function is shown in Fig. \ref{plot:2n-r2-0.3}. For the simulations,
buffer sizes are fixed at $14$ and $X^1, \, X^2, \, Y$ are distributed according to the Poisson distribution
with $\E Y = 13 \text{ and } \E {X^2} = 1.0$.

In Fig. \ref{plot:2n-r2-0.3}, the x-axis indicates the change in data rate of Node 1. 
This setup is akin to that used in Fig. \ref{plot:2n-demax14}. The y-axis indicates the normalized average queue length
of all the nodes. We considered values $r_1 = 0.7$ and $r_2 = 0.3$.
The plot indicates only the average queue length of all nodes, since our objective is to minimize
the average delay of transmission of data (which is related to the data queue length).
From Fig. \ref{plot:2n-r2-0.3}, it can be observed that all learning algorithms show an increase in the 
collective average queue length (referred to as the normalized average cost in Figs. \ref{plot:2n-demax14}-\ref{plot:4n-sa-levels}).
This occurs because by using the cost function \eqref{eqn:cost-function-with-r2} the learning algorithms 
(Q-learning with UCB and $\epsilon$-greedy exploration as well as combined nodes Q-learning) give less importance
to the queue length component in the cost function. Thus the policies learnt by these algorithms minimize the energy usage albeit with an increase
in data queue length. As the figure shows, the learning algorithms we described in Section \ref{sec:2nodes-algo} perform much better
compared to the greedy and combined nodes methods.
\begin{figure} 
    \begin{subfigure}[t]{0.5\columnwidth}
      \begin{center}
       \scalebox{0.49}{
%
%
%
\definecolor{mycolor1}{rgb}{0.04314,0.51765,0.78039}%
\definecolor{mycolor2}{rgb}{0.00000,0.49804,0.00000}%
\begin{tikzpicture}

\begin{axis}[%
width=4.7715625in,
height=4.3076875in,
scale only axis,
xmin=0,
xmax=2.5,
xtick={0,0.25,0.5,0.75,1.0,1.25,1.5,1.75,2.0,2.25,2.5},
xlabel={$\E {X^1}$},
ymin=0.4,
ymax=2.6,
ylabel={Normalized Average Queue Length of all nodes},
legend style={draw=black,fill=white,legend cell align=left}
]

\addplot [color=red,solid,line width=2.0pt,mark=diamond,mark options={solid},mark size=3pt] table[row sep=crcr]
{
0.0     0.460701\\
0.25	0.50831\\
0.5	0.58602\\
0.75	0.7306\\
1	0.99164\\
1.25	1.32357\\
1.5	1.54487\\
1.75	1.6466\\
2	1.6944\\
2.25	1.71966\\
2.5	1.73438\\
};
\addlegendentry{Greedy};

\addplot [color=mycolor2,solid,line width=2.0pt] table[row sep=crcr]
{
0.0     0.50999
0.25	0.598618\\
0.5	0.801197\\
0.75	1.0036\\
1	1.56\\
1.25	1.78\\
1.5	1.9998\\
1.75	2.0\\
2	2.0\\
2.25	2.0\\
2.5	2.0\\
};
\addlegendentry{Combined Nodes QL};

\addplot [color=blue,solid,line width=2.0pt,mark=o,mark options={solid},mark size=3pt] table[row sep=crcr]
{
0.0	0.46013\\
0.25	0.49675\\
0.5	0.50139\\
0.75	0.540141\\
1	0.631084\\
1.25	0.702871\\
1.5	0.896241\\
1.75	1.09465\\
2	1.352\\
2.25	1.40175\\
2.5	1.43031\\
};
\addlegendentry{QL $\epsilon$-greedy};

\addplot [color=black,dash pattern=on 1pt off 3pt on 3pt off 3pt,line width=2.0pt]
  table[row sep=crcr]{
0.0     0.46003\\
0.25	0.47241\\
0.5	0.48349\\
0.75	0.50165\\
1	0.55192\\
1.25	0.59674\\
1.5	0.74562\\
1.75	0.98326\\
2	1.23519\\
2.25	1.35624\\
2.5	1.37819\\
};
\addlegendentry{Q-UCB};

\end{axis}
\end{tikzpicture}%
}
	\caption{$r_1 = 0.7, \,r_2 = 0.3$}
	\label{plot:2n-r2-0.3}
	\end{center}
    \end{subfigure}%
    \begin{subfigure}[t]{0.5\columnwidth} 
	\begin{center}
        \scalebox{0.49}{
%
%
%
\definecolor{mycolor1}{rgb}{0.04314,0.51765,0.78039}%
\definecolor{mycolor2}{rgb}{0.00000,0.49804,0.00000}%
\definecolor{mycolor3}{rgb}{0.8,0.3,0.1}%
\begin{tikzpicture}

\begin{axis}[%
width=4.7715625in,
height=4.3076875in,
scale only axis,
xmin=0,
xmax=2.5,
xtick={0,0.25,0.5,0.75,1.0,1.25,1.5,1.75,2.0,2.25,2.5},
xlabel={$\E {X^1}$},
ymin=0.4,
ymax=2.3,
ylabel={Normalized Average Cost},
legend style={draw=black,fill=white,legend cell align=left}
]

\addplot [color=red,dash pattern=on 1pt off 3pt on 3pt off 3pt,line width=2.0pt] table[row sep=crcr]
{
0.0     0.460701\\
0.25	0.50831\\
0.5	0.58602\\
0.75	0.7306\\
1	0.99164\\
1.25	1.32357\\
1.5	1.54487\\
1.75	1.6466\\
2	1.6944\\
2.25	1.71966\\
2.5	1.73438\\
};
\addlegendentry{Greedy};

\addplot [color=mycolor2,solid,line width=2.0pt,mark=+,mark options={solid},mark size=4pt] table[row sep=crcr]
{
0.0     0.454837
0.25	0.487531\\
0.5	0.56333\\
0.75	0.643118\\
1	0.88301\\
1.25	0.885\\
1.5	1.17746\\
1.75	1.41471\\
2	1.53603\\
2.25	1.59425\\
2.5	1.62519\\
};
\addlegendentry{Combined Nodes QL};

\addplot [color=blue,solid,line width=2.0pt,mark=o,mark options={solid},mark size=3pt] table[row sep=crcr]
{
0.0	0.452877\\
0.25	0.4684\\
0.5	0.487809\\
0.75	0.5091\\
1	0.52\\
1.25	0.601472\\
1.5	0.817782\\
1.75	1.08151\\
2	1.28094\\
2.25	1.34175\\
2.5	1.35916\\
};
\addlegendentry{QL $\epsilon$-greedy};

\addplot [color=mycolor3,solid,line width=2.0pt] table[row sep=crcr]
{
0.0	0.454558\\
0.25	0.469\\
0.5	0.5392\\
0.75	0.57528\\
1	0.773157\\
1.25	0.83985\\
1.5	1.0728\\
1.75	1.18705\\
2	1.36782\\
2.25	1.39235\\
2.5	1.53643\\
};
\addlegendentry{QL-SA $\epsilon$-greedy};

%
%

\end{axis}
\end{tikzpicture}%
}
        \caption{$r_1 = 1$}
        \label{plot:2n-ql-qlsa-gr-cnql}
	\end{center}
    \end{subfigure}
    \caption{Performance comparison of policies: $E_{\!{MAX}} = D_{\!{MAX}} = 14$ when 
	      $X^1,\,X^2,\,Y$ are Poisson distributed with $\E Y = 13, \E {X^2} = 1.0$ }   
\end{figure}
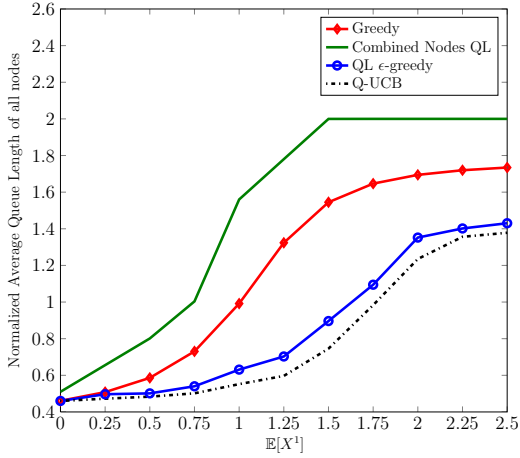
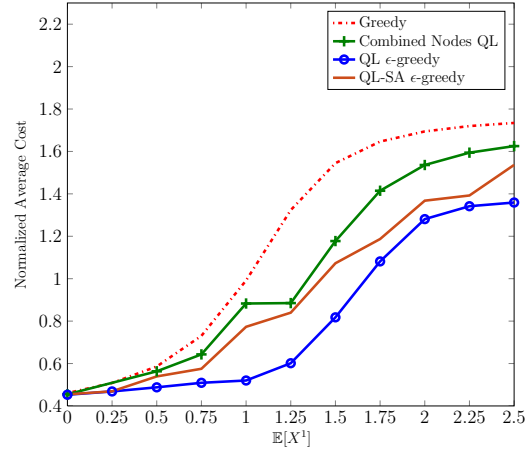

In Fig. \ref{plot:2n-ql-qlsa-gr-cnql}, the performance comparison of Q-learning with and without state aggregation is 
shown for the case of two nodes and an EH source (i.i.d case) and compared with greedy and combined nodes Q-learning
mathod. The $\epsilon$-greedy exploration mechanism is 
used for both algorithms.
The experimental setup is similar to that used in Fig. \ref{plot:2n-demax14}. 
The x-axis indicates the variation in data rate of Node 1, while the y-axis indicates the normalized 
average cost of the nodes. The algorithm in Section \ref{sec:approx-alg} was simulated by 
partitioning the data and energy buffers into 3 partitions each.
It can be observed in Fig. \ref{plot:2n-ql-qlsa-gr-cnql} that Q-learning with state aggregation performs better than the greedy and combined nodes methods. However since Q-learning with state aggregation algorithm finds 
near-optimal policy, its performance is not as good as the algorithm in Section \ref{sec:ql} with the same exploration mechanism.
\begin{remark}
 The Greedy algorithm distributes the available energy among the sensor nodes based on the proportion
of data available in the nodes. It shares all the available energy at every decision instant without storing it for future use.
We compare our algorithms with the Greedy algorithm in order to show that myopic strategy may not be optimal. 
Our results show that one has to devise the policy not only for the present requirement for energy but also for the 
future energy requirements as well. This idea is naturally incorporated in our RL algorithms. 
Moreover, Greedy policy is optimal when the conversion function $g$ is linear. This has been derived in \cite{sharma2010optimal}
for the case of single sensor. The performance of the algorithm proposed in \cite{prabuchandran2013q} with non-linear $g$ is compared
with the performance of the greedy method. Thus, the comparison of the performance of our algorithms with the greedy method
also follows naturally from the earlier cited works.

The Combined Nodes Q-learning method learns the policy which maps the total number of data bits available 
in all the nodes to the total amount of energy required. The energy sharing between the nodes is then based on the 
proportion of data available in the nodes. Under the Combined Nodes Q-learning algorithm, the state space is greatly reduced, 
i.e., instead of the cartesian product of states in each node (as in our Q-learning method with and without state aggregation), 
it is just the sum of the states of the sensor nodes. So, the learning is faster in combined nodes Q-learning algorithm. 
However, the policy learnt is suboptimal as was shown in Figs. \ref{plot:2n-demax14}-\ref{plot:2n-ql-qlsa-gr-cnql} and performs poorly in 
comparison with our algorithms. So, we compare our algorithms with Combined nodes Q-learning to illustrate the tradeoff 
of size of the state space with the nature of the obtained policy.

Note that our RL algorithms learn the energy sharing policy not quantized to a single point but considers energy sharing 
among the sensor nodes. Learning an optimal energy sharing scheme is a difficult problem. Hence, we would like to understand 
how well our algorithms perform against a simple heuristic policy such as Greedy or a policy obtained from the Combined nodes Q-learning method.
\end{remark}

\begin{remark}
\label{remark:different-g}
The function $g(.)$ gives the number of bits that can be transmitted using certain units of energy. 
Our algorithms work regardless of the forms of $g$.
RL algorithms use the simulation samples to learn the energy sharing policy by trying out various actions in each of the states. In our problem, at time $k$ let us assume we are in state $s_k = (q_{k}^1,q_{k}^2,\ldots,q_{k}^n,E_k,X_{k-1},Y_{k-1})$, i.e., the data in the data buffer and energy in the energy buffer are fixed to some values. Based on the current Q-value, we share the energy available to the various sensor nodes by selecting action $T_k=(T_k^1,T_k^2,\ldots, T_k^n)$. Depending on the
action $T_k$, the state of the system evolves according to \eqref{eqn:qk-evolution}-\eqref{eqn:ek-evolution}.

In order to find the next state of the system (\eqref{eqn:qk-evolution} - \eqref{eqn:ek-evolution}), 
it suffices to know the number of bits that got transmitted by chosing the action $T_k$ in 
slot $k$ in a real system, which is given by $g(T_k)$. It must be noted that we do not need information on 
the functional form of $g$ for finding the next state, but only the value of the function for action $T_k$.
This value can be observed (in a real system) even if we do not have the precise
model for the Gaussian function in terms of $g(\cdot)$. In other words, all we
need is to observe the number of bits that got transmitted by supplying $T_k$ units of energy.

To update the Q-value of state-action pair $(s_k,T_k)$ (see \eqref{eqn:q-learning-update}), 
we need to know the cost $c(s_k,T_k)$ incurred by choosing action $T_k$ in state $s_k$, which is computed
using \eqref{eqn:cost-fn}, where again we only require information on $g(T^i_k), ~i=1,2, \ldots n$, 
but not the exact form of $g(\cdot)$.
Our proposed RL algorithms work by updating Q values, and such an updation essentially requires the 
cost information (computed using \eqref{eqn:cost-fn}).
Similarly in the cross entropy method, to compute the average cost of the policy, we need to 
compute the single-stage cost (using \eqref{eqn:cost-fn}). In summary, our algorithms do 
not require the exact form of $g(\cdot)$.

In the case of the greedy algorithm, in order to decide the number of energy units $T_k$ that need to be shared, the function $g^{-1}(\cdot)$ and hence the functional form of $g(\cdot)$ must be known (see Section \ref{sec:results}), 
i.e., one needs to obtain the mathematical model for the conversion function. In comparison, as 
stated before, our algorithms do not need such information.

However, to simulate the environment, we need to know the functional form of the conversion function $g$. But, 
in a real physical system, our algorihms do not require the functional form of $g$. Figure \ref{plot:2n-different-g}
illustrates the performance of our algorithms and the Greedy and Combined nodes Q-learning methods
for a different form of function $g(\cdot)$, i.e., $g(\cdot)= \sqrt{3\log(1+x)}$. 
The setup is similar to that of Fig. \ref{plot:2n-non-iid}. 
We observe from Fig. \ref{plot:2n-different-g} that irrespective of the form of $g(\cdot)$, 
our algorithms find good policies, since they do not require this knowledge to do so.
\end{remark}
\begin{figure}[h]
 \begin{center}
   \scalebox{0.7}{
%
%
%
\definecolor{mycolor1}{rgb}{0.04314,0.51765,0.78039}%
\definecolor{mycolor2}{rgb}{0.00000,0.49804,0.00000}%
\begin{tikzpicture}

\begin{axis}[%
width=4.7715625in,
height=4.3076875in,
scale only axis,
xmin=0,
xmax=10,
xtick={0,1,2,3,4,5,6,7,8,9,10},
xlabel={$\E{\omega_1}$},
ymin=0.000000115,
ymax=1.7,
ylabel={Normalized Average Cost},
legend style={draw=black,fill=white,legend cell align=left}
]

\addplot [color=red,solid,line width=2.0pt,mark=diamond,mark options={solid},mark size=3pt] table[row sep=crcr]
{
0        0.00211652  	\\
1	 0.017277  	\\
2	 0.336894	\\
3	 0.990259 	\\
4	 1.17565  	\\
5	 1.22646   	\\
6	 1.24401   	\\
7	 1.25081   	\\
8	 1.25337   	\\
9	 1.25455    	\\
10	 1.25497  	\\
};
\addlegendentry{Greedy};

\addplot [color=mycolor2,solid,line width=2.0pt] table[row sep=crcr]
{
0  	 0.00213488    	\\
1	 0.0173912   	\\
2	 0.35141   	\\
3	 1.01313    	\\
4	 1.1692    	\\
5	 1.21697     	\\
6	 1.23108     	\\
7	 1.23505     	\\
8	 1.23984     	\\
9	 1.25305     	\\
10	 1.26717      	\\
};
\addlegendentry{Combined Nodes QL};
\addplot [color=blue,solid,line width=2.0pt,mark=o,mark options={solid},mark size=3pt] table[row sep=crcr]
{
0 	0.001923      	\\
1	0.0171		\\
2	0.301456	\\
3	0.956789	\\
4	1.14123		\\
5	1.16145		\\
6	1.19666		\\
7	1.20342		\\
8	1.21116		\\
9	1.22598 	\\
10	1.23567		\\
};
\addlegendentry{QL $\epsilon$-greedy};
\addplot [color=black,dash pattern=on 1pt off 3pt on 3pt off 3pt,line width=2.0pt]
  table[row sep=crcr]{
0 	0.0019		\\
1	0.0161		\\
2	0.23		\\
3	0.85		\\
4	1.02		\\
5	1.12		\\
6	1.134		\\
7	1.14		\\
8	1.154		\\
9	1.167 		\\
10	1.2		\\
};
\addlegendentry{Q-UCB};

\end{axis}
\end{tikzpicture}%
}
    \caption{$E_{\!{MAX}} = 20,\,D_{\!{MAX}} = 10$, $\omega_1,\,\omega_2,\,\chi$ are Poisson distributed 
		 with $\E {\omega_2} = 1.0$,$\E {\chi} = 20$, $g(x) = \sqrt{3\log(1+x)}$ } 
     \label{plot:2n-different-g}
 \end{center}
\end{figure}
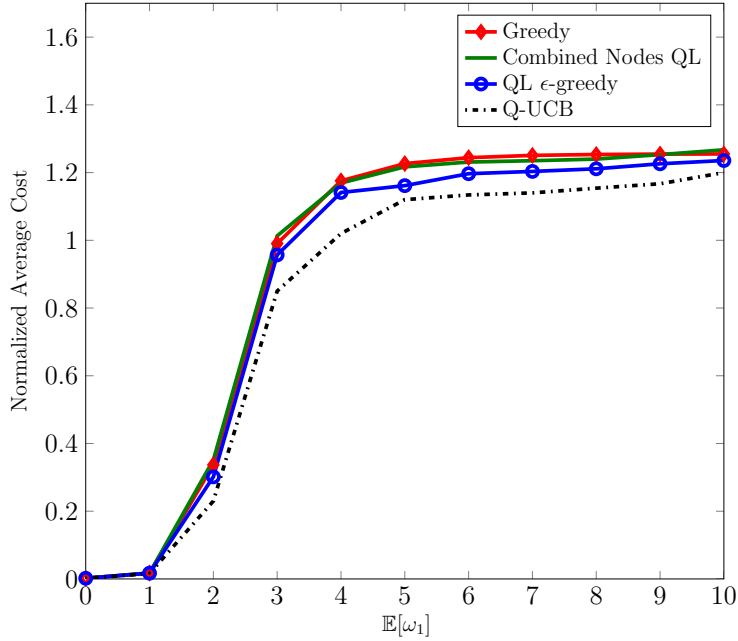
\section{Conclusions and Future Work}
\label{sec:conclusions}
We studied the problem of energy sharing in sensor networks and proposed a new technique to manage energy available through harvesting.
Multiple nodes in the network sense random amounts of data and share the energy harvested by an
energy harvesting source.
We presented an MDP model for this problem and an algorithm that determines the optimal amount of energy to be 
supplied to every node at a decision instant. The algorithm minimizes the sum of (data) queue lengths in the data buffers, by finding
the optimal energy split profile. In order to deal with the curse of dimensionality, we also proposed approximation algorithms
that employ state aggregation effectively to reduce the computational complexity.
Numerical experiments  showed that our algorithms outperform the algorithms described in Section \ref{sec:results}.

Our future work would involve applying threshold tuning for state aggregation, gradient based approaches and basis adaptation methods 
for policy approximation. 
The partitions formed for clustering the state space (Section \ref{sec:t-features}) can be improved
by tuning the partition thresholds (see \cite{6247516}). This method can be employed to obtain improved
deterministic policies when state-action space is extremely large.
Gradient based methods \cite{bhatnagar2013stochastic}, \cite{konda2003onactor}, \cite{bhatnagar2009natural} approximate 
the policy using parameter $\theta$ and a set of given (fixed) basis functions $\{f_k : 1 \leq k \leq n\}$. 
Typically a probability distribution over the actions corresponding to a state is defined using $\theta$ and $\{f_k\}$. 
The parameter is updated using the gradient direction of the policy performance, which is usually the long-run average or discounted cost of the policy.
In the approximation algorithm described in Section \ref{sec:cross-entropy}, the basis functions used 
in the policy parameterization are fixed. One could obtain better policies if the basis functions are also optimized.
Basis adaptation methods \cite{menache2005basis}, \cite{bhatnagar2013feature} start with a given set of basis functions. 
The random policy parameter $\theta$ is updated using simulated trajectories of the MDP on a faster timescale. 
The basis functions are tuned on a slower timescale. These methods can be employed to find better policies.
We shall also develop prototype implementations for this model and test our algorithms.

\section*{Acknowledgements}
The authors would like to thank all the three reviewers of \cite{manuscript} for their detailed comments that significantly
helped in improving the quality of this report and the manuscript \cite{manuscript}. This work was supported in part through projects from the Defence Research and Development Organisation (DRDO) and the Department of Science and Technology (DST),
Government of India.
\bibliographystyle{plain}
\bibliography{arxiv_version}

\end{document}